\documentclass[a4paper]{article}
\usepackage{fullpage}

\usepackage[utf8]{inputenc} 
\usepackage[T1]{fontenc}
\usepackage{url}
\usepackage{ifthen}
\usepackage{cite}
\usepackage[cmex10]{amsmath}
\usepackage{graphicx}
\usepackage{subcaption}
\usepackage{marginnote}
\usepackage{lipsum}
\usepackage{graphicx}
\usepackage[utf8]{inputenc}
\usepackage{tikz}
\usepackage{amsthm}
\usepackage{amsfonts}
\usepackage{cite}
\usepackage{color}
\usepackage{comment}

\newtheorem{theorem}{Theorem}[section]
\newtheorem{corollary}{Corollary}[theorem]
\newtheorem{rem}{Remark}[theorem]
\newtheorem{lemma}{Lemma}[theorem]
\theoremstyle{definition}
\newtheorem{definition}{Definition}[section]
\newcommand{\x}{\mathbf{x}}
\newcommand{\y}{\mathbf{y}}
\newcommand{\A}{\mathbf{A}}
\newcommand{\w}{\mathbf{w}}
\newcommand{\X}{\mathbf{X}}
\newcommand{\ramp}{\mathbf{r}}
\newcommand{\Z}{\mathbf{Z}}
\newcommand{\z}{\mathbf{z}}

\MakeRobust\widetilde
\begin{document}
\title{Rigorous State Evolution Analysis for Approximate Message Passing with Side Information} 
\author{
Hangjin Liu, \textit{Student Member, IEEE}, Cynthia Rush,\textit{ Member, IEEE}, \\and Dror Baron\textit{ Senior Member, IEEE}\thanks{
A subset of this work has been presented at the IEEE International Symposium on Information Theory,
July 2019 \cite{Liu2019ISIT}.
The work of Liu and Baron was supported in part by NSF EECS $\#1611112$, and that of Rush by NSF CCF $\#1849883$.
Liu and Baron are with the Electrical and Computer Engineering Department at North Carolina State University,
Raleigh, NC, email $\{$hliu25, barondror$\}$@ncsu.edu. Rush is with the Department of Statistics at Columbia University, New York, NY email cynthia.rush@columbia.edu.
}
}

\date{\vspace{-5ex}}

\maketitle

\begin{abstract}
A common goal in many research areas is to reconstruct an unknown signal $\x$ from noisy linear measurements.  Approximate message passing (AMP) is a class of low-complexity algorithms that can be used for efficiently solving such high-dimensional regression tasks.  Often, it is the case that side information (SI) is available during reconstruction.  For this reason, a novel algorithmic framework that incorporates SI into AMP, referred to as approximate message passing with side information (AMP-SI), has been recently introduced. In this work, we provide rigorous performance guarantees for AMP-SI when there are statistical dependencies between the signal and SI pairs and the entries of the measurement matrix are independent and identically distributed (i.i.d.) Gaussian.  The AMP-SI performance is shown to be provably tracked by a scalar iteration referred to as state evolution (SE).  Moreover, we provide numerical examples
that demonstrate empirically that the SE can predict the AMP-SI mean square error accurately. 
\end{abstract}

\section{Introduction}
High-dimensional linear regression is a well-studied model used in many applications including compressed sensing\cite{DMM2009}, 
imaging\cite{Arguello2011}, and machine learning and statistics\cite{Hastie2001}. The unknown signal $\x \in\mathbb{R}^n$ is viewed through the linear model:
\begin{equation}\label{eq:1-1}
\y=\A \x+ \w,
\end{equation}
where $\mathbf{y}\in\mathbb{R}^m$ are the measurements, $\A\in\mathbb {R}^{m\times n}$ is a known measurement matrix, and  $\w\in \mathbb R^m$ is measurement noise. The goal is to estimate the unknown signal $\x$ having knowledge only of the noisy measurements $\y$ and the measurement matrix $\A$.  When the problem is under-determined (i.e., $m<n$), in order for reconstruction to be successful, it is necessary to exploit structural or probabilistic characteristics of the input signal $\x$.  Often a prior distribution on the input signal $\x$ is assumed, and in this case approximate message passing (AMP) algorithms\cite{DMM2009} can be used for the reconstruction task.

AMP~\cite{DMM2009, RanganGAMP2010} is a class of low-complexity algorithms that can be used for efficiently solving high-dimensional regression tasks (\ref{eq:1-1}). AMP works by iteratively generating estimates of the unknown input vector,  $\x$, using a possibly non-linear denoiser function tailored to any prior knowledge about $\x$. 
One favorable feature of AMP is that under some technical conditions on the measurement matrix $\A$ and 
signal
$\x$, the observations at each iteration of the algorithm are almost surely equal in distribution to $\x$ plus independent and identically distributed (i.i.d.)\,\,Gaussian noise in the large system limit, i.e., $m, n \rightarrow \infty$ with ${m}/{n} \rightarrow \delta \in (0, \infty)$.

\textbf{AMP with Side Information (AMP-SI):} In information theory~\cite{Cover06}, when different communication systems share side information (SI), overall communication can become more efficient. Recently~\cite{Dror2017, Ma2018}, 
an algorithmic framework dubbed AMP-SI was introduced that incorporates
SI into AMP for high-dimensional regression tasks (\ref{eq:1-1}).  AMP-SI has been empirically demonstrated to have good reconstruction quality and is easy to use. For example, it has been proposed to use AMP-SI for channel estimation in emerging millimeter wave communication systems~\cite{saleh1987statistical}, where the time dynamics of the channel structure allow previous channel estimates to be used as SI when estimating the current channel structure~\cite{Ma2018}.

In this work, we model the observed SI, denoted by ${{\widetilde{\x}}\in\mathbb{R}^n}$,  as depending statistically on the unknown signal $\x$ through some joint probability density function (pdf), $f({\X},\widetilde{{\X}})$. AMP-SI uses a conditional denoiser, $g_t:\mathbb{R}^{2n}\rightarrow \mathbb{R}^n$, to incorporate SI,
\begin{equation}\label{eq:eta_2}
g_t(\mathbf{a}, \mathbf{b})=\mathbb{E}[\X | \X+\lambda_t\mathcal{N}(\mathbf{0},
\pmb{\mathbb{I}}_n)=\mathbf{a}, \widetilde \X= \mathbf{b}],
\end{equation}
where $\mathcal{N}(\boldsymbol \mu,\boldsymbol\Sigma)$ is a Gaussian random vector with mean vector $\boldsymbol \mu$ and $\boldsymbol \Sigma$, and $\pmb{\mathbb{I}}_n$ denotes
an identity matrix of size $n \times n$.  The constant $\lambda_t$ used in \eqref{eq:eta_2} is determined by the state evolution introduced below.

The AMP-SI algorithm iteratively updates estimates of the input signal $\x$: let $\x^0  = \mathbf{0}$, the all-zeros vector, then
\begin{align}
\ramp^t &=\y-\A\x^t+\frac{\ramp^{t-1}}{\delta} [ \text{div} \, g_{t-1}(\x^{t-1}+\A^T\ramp^{t-1}, \widetilde \x)], \label{eq:1-5} \\
 \x^{t+1}&={g_{t}}(\x^{t}+\A^T\ramp^{t}, \widetilde \x), \label{eq:1-6}
\end{align}
where 
$\x^t\in \mathbb{R}^n$ is the estimate of $\x$ at iteration $t$, the measurement rate is $\delta={m}/{n}$, and
$\ramp^t$ is the residual or unexplained part of the measurements. The term $\x^{t}+\A^T\ramp^{t}$ is known as pseudo data and in the non-SI case has been proved
to be equal in distribution to the input $\x$ plus white Gaussian noise in the large system limit. For a differential function, 
$g: \mathbb{R}^{2n} \rightarrow \mathbb{R}^n$, we let $\text{div} \, g(\mathbf{a}, \mathbf{b})$ be the divergence of the function with respect to (w.r.t.) the first argument, namely,
we use \begin{equation}
\text{div}\, g(\mathbf{a}, \mathbf{b}) = \sum_{i=1}^n \frac{\partial g_i}{\partial a_i}(\mathbf{a}, \mathbf{b}).\end{equation}The AMP-SI algorithm~\eqref{eq:1-5}--\eqref{eq:1-6} uses the conditional denoiser \eqref{eq:eta_2}.

\textbf{State Evolution (SE):} It has been proven that the performance of AMP, as measured, for example, by the normalized squared $\ell_2$-error $\frac{1}{n}||\x^t-\x||_2^2$ between the estimate  $\x^t$ and true signal $\x$, can be accurately  predicted by a scalar recursion referred as SE\cite{Bayati2011,Rush_Finite18}. SE analyses require that the matrix $\A$ be
i.i.d.\ Gaussian and various assumptions on the elements of the signal. 
The SE equation for AMP-SI is as follows. Assume the entries of the noise $\w$ are i.i.d.\ $\sim f(W)$ with $\sigma_w^2 = \mathbb{E}[W^2]$, and let $\lambda_0^2 = \sigma_w^2 + \lim_m ||\x||^2/m$. Note that $\lambda_0^2$ is the initial variance of the difference between the pseudo-data at iteration $t=0$, i.e.\ $\x^{0}+\A^T\ramp^{0} = \A^T \y$, and signal $\x$,
\begin{equation}
\lambda_t^2 = \sigma_w^2 +  \lim_m \frac{1}{m}\mathbb{E}_{\Z}\left[||g_{t-1}(\x + \lambda_{t-1}\Z, \widetilde{\x}) - \x||^2\right],
\label{eq:SE2}
\end{equation}
where
$\Z\sim \mathcal{N}(\mathbf{0},\pmb{\mathbb{I}}_n)$. In particular, the value of $\lambda_t$ defined above is used to define the conditional denoiser in \eqref{eq:eta_2} to be used in the AMP-SI algorithm in~\eqref{eq:1-5}--\eqref{eq:1-6}.

Considering AMP-SI~\eqref{eq:1-5}-\eqref{eq:1-6},  however, we cannot directly apply the existing AMP theoretical results~\cite{Bayati2011, Rush_Finite18}, as the conditional denoiser in~\eqref{eq:eta_2} is not a separable denoiser, in that its output at any index $i$ may depend on all other indices of the input.  Even when the signal and SI pair, $(\x, \widetilde{\x})$, is sampled i.i.d.\ 
from the joint pdf
$f(X, \widetilde{X})$, each entry of the signal $\x$ is generated
according to a different conditional density depending on the corresponding SI $\widetilde{x}_i$, and therefore the signal $\x$ now has independent, but not identically
distributed entries.  
The usual results~\cite{Bayati2011, Rush_Finite18} require that the same
denoiser be applied to each entry of the pseudo-data. Therefore, even in the case of an i.i.d.\ distributed signal and SI pair, the theoretical results in~\cite{Bayati2011,Rush_Finite18} do not apply, because different scalar denoisers are needed for each index.
Recent results~\cite{Berthier2017}, however, extend the asymptotic SE analysis to a larger class of possible denoisers, allowing, for example, each element of the input to use a different non-linear denoiser as is the case in AMP-SI.  We employ these results to rigorously relate the SE presented in \eqref{eq:SE2} to the AMP-SI algorithm in~\eqref{eq:1-5}-\eqref{eq:1-6}.

{\bf Related Work:}
While integrating SI into reconstruction algorithms is not new, AMP-SI introduces a unified framework within AMP supporting arbitrary signal and SI dependencies.
Prior work using SI has been either heuristic, limited 
to specific applications, or outside the AMP framework. 

For example, Wang and Liang~\cite{wang2015approximate} integrate SI into AMP for a specific signal prior density, but the method is difficult to apply to other signal models.  Ziniel and Schniter~\cite{DCSAMP} develop an AMP-based reconstruction algorithm for a
time-varying signal model based on Markov processes for the support and amplitude.  This signal model is easily incorporated into the AMP-SI framework as discussed in the analysis of the birth-death-drift model of~\cite{Dror2017, Ma2018}.  Manoel \emph{et al}.\ implement an AMP-based algorithm in which
the input signal is repeatedly reconstructed in a streaming fashion,
and information from past reconstruction attempts is aggregated into a prior,
thus improving ongoing reconstruction results~\cite{manoel2017streaming}.
This reconstruction scheme resembles that of AMP-SI, in particular when the Bernoulli-Gaussian model is used (see Section~\ref{subsec:examples iid}).

{\bf Contribution and Outline:} Ma \emph{et al}.~\cite{Ma2018} use numerical experiments to show that SE~\eqref{eq:SE2} accurately tracks the performance of AMP-SI~\eqref{eq:1-5}-\eqref{eq:1-6}, as was shown rigorously for standard AMP, and they conjecture that rigorous theoretical guarantees can be given for AMP-SI as well. In this work, we analyze AMP-SI performance when the input signal and SI are drawn according to a general pdf $f(\X, \widetilde{\X})$ obeying some finite moment conditions, the AMP-SI denoiser~\eqref{eq:eta_2} is uniformly Lipschitz, and the measurement matrix $\A$ is i.i.d.\ Gaussian.

In Section~\ref{main_result}, we provide the main results. Examples for various signal and SI models, and numerical experiments comparing the empirical performance of AMP-SI and the SE predictions, are provided in Section~\ref{examples}. The technical proofs of our main results are given in Sections~\ref{main_proof} and~\ref{main_proof_non-separable}.

{\bf Notation:}
Throughout the work we will use bold symbols to indicate vectors and non-bold to indicate scalars.  Capital letters will indicate random variables (RVs), and lowercase letters their realizations.  For example, $\widetilde{\X}$ is a random vector of SI, $\widetilde{\x}$ is a realization,  and $\widetilde{x}_i$ is the $i^{th}$ entry of the realization.  With a slight abuse of notation we also use bold, capital letters to indicate random matrices, like the measurement matrix $\A$.  We let $||\cdot||$ denote the  Euclidean norm, and $\overset{p}{=}$ denotes convergence in probability. The expectation $\mathbb{E}_\Z$ represents the expected value w.r.t.\  $\Z$.

\section{Main Results}\label{main_result}
Our main results provide AMP-SI performance guarantees when considering \textit{pseudo-Lipschitz} loss functions, defined below.

\begin{definition}
\label{def:PLfunc}
\textbf{\emph{Pseudo-Lipschitz functions}}~\cite{Berthier2017}: For $k\in \mathbb{N}_{>0}$ and any $n,m\in \mathbb{N}_{>0}$, a function $\phi : \mathbb{R}^n\to
\mathbb{R}^m$ is \emph{pseudo-Lipschitz of order $k$}, or \emph{PL(k)}, if there exists a constant $L$, referred to as the pseudo-Lipschitz constant of $\phi$, such that for $\x, \y \in \mathbb{R}^n$,
\begin{equation*}
\frac{\left|\left|\phi(\x)-\phi(\y)\right|\right|}{\sqrt{m}}\leq L\Big(1+ \Big(\frac{||\x||}{\sqrt{n}}\Big)^{k-1}+ \Big(\frac{||\y||}{\sqrt{n}}\Big)^{k-1}\Big) \frac{||\x-\y||}{\sqrt{n}}.
\end{equation*}
For $k = 1$, this definition coincides with the standard definition of a Lipschitz function. %

A sequence (in $n$) of PL(k) functions $\{\phi_n\}_{n\in \mathbb{N}_{>0}}$
is called \emph{uniformly pseudo-Lipschitz}
of order $k$, or \emph{uniformly PL(k)}, if, denoting by $L_n$ the pseudo-Lipschitz constant of $\phi_n$, we have $L_n < \infty$ for each $n$ and $\lim\sup_{n\to\infty}L_n < \infty$. We refer to a function as \emph{uniformly Lipschitz} if it is uniformly PL(1).
\end{definition}
Throughout the work we will use the following result about pseudo-Lipschitz functions.  Its proof can be found in Appendix~\ref{app:technical}.

\begin{lemma}
\label{prop_Lipschitz}
For any PL(k) function $\phi : \mathbb{R}^n\to
\mathbb{R}^m$ with PL constant $L > 0$, there exists a constant $L' = \max\{2L, ||\phi(\mathbf{0})||/\sqrt{m}\} > 0$ where $\mathbf{0} \in \mathbb{R}^n$ is a vector of zeroes, such that for $\x\in\mathbb{R}^n$, 
\begin{equation*} 
\frac{\left|\left|\phi(\x)\right|\right|}{\sqrt{m}}
\leq L'\Big(1+ \Big(\frac{||\x||}{\sqrt{n}}\Big)^k\Big).
\end{equation*}
\end{lemma}

In Section~\ref{sec:sep} we present theoretical results for the simplified situation where the signal and the SI have i.i.d.\ entries, meaning that separable denoisers can be used. Though the separable denoiser is a simplified version of the main result, we believe that stating this case will elucidate the main technical pieces used in the more general non-separable case. In Section~\ref{sec:non-sep} we provide theoretical analysis of the general signal and SI model using non-separable denoisers in the AMP-SI iteration.

\subsection {Separable denoiser} 
\label{sec:sep}
In the case of $(\x, \widetilde{\x})$ sampled i.i.d.\ 
from the joint pdf
$f(X, \widetilde{X})$, the conditional denoiser of AMP-SI \eqref{eq:eta_2}
is separable.  Define $\eta_t: \mathbb{R}^2 \rightarrow \mathbb{R}$ as
\begin{equation}\label{eq:eta_2_iid}
\eta_t(a, b)=\mathbb{E}[X | X+\lambda_t\mathcal{N}(0,
1)=a, \widetilde X= b],
\end{equation}
and the AMP-SI algorithm in \eqref{eq:1-5}-\eqref{eq:1-6} simplifies to
\begin{align}
&\ramp^t =\y-\A\x^t+\frac{\ramp^{t-1}}{\delta} \sum_{i=1}^n \eta'_{t-1}([\x^{t-1}+\A^T\ramp^{t-1}]_i, \widetilde x_i), \label{eq:1-5_iid} \\
& x_i^{t+1} ={\eta_{t}}([\x^{t}+\A^T\ramp^{t}]_i, \widetilde x_i), \quad \text{ for } i = 1, 2, \ldots, n, \label{eq:1-6_iid}
\end{align}
where the derivative $\eta_t'(s, \cdot) = \frac{\partial}{\partial s} \eta_t(s, \cdot)$. We highlight that the difference between these equations and the standard (no SI) AMP equations is that the denoiser used in AMP-SI incorporates the SI. Owing to $\x$ no longer being i.i.d., the conditional denoiser \eqref{eq:eta_2_iid} depends
on the index $i$, meaning that different scalar denoisers will be used at different indices,
based on different SI values at different indices.  For the denoiser in \eqref{eq:eta_2_iid}, the SE is as follows: let $\lambda_0^2 = \sigma_w^2 + \mathbb{E}[X^2]/\delta$ and for $t \geq 1$,
\begin{equation}
\lambda_t^2 = \sigma_w^2 + \frac{1}{\delta}\mathbb{E}\left[(\eta_{t-1}(X + \lambda_{t-1}Z, \widetilde{X}) - X)^2\right],
\label{eq:SE2_iid}
\end{equation}
where $(X, \widetilde{X}) \sim f(X, \widetilde{X})$ are independent of $Z\sim \mathcal{N}(0,1)$.

The following result characterizes the asymptotic performance of separable AMP in \eqref{eq:1-5_iid}-\eqref{eq:1-6_iid}
when the performance is measured with pseudo-Lipschitz loss.

\begin{theorem}
\label{thm:SE}
For any PL(2) functions, $\phi: \mathbb{R}^{2} \rightarrow \mathbb{R}$ and $\psi: \mathbb{R}^{3} \rightarrow \mathbb{R}$, 
define sequences of  functions, $\phi_m: \mathbb{R}^{2m} \rightarrow \mathbb{R}$ and $\psi_n: \mathbb{R}^{3n} \rightarrow \mathbb{R}$, 
as follows: for vectors $\mathbf{a}, \mathbf{b} \in \mathbb{R}^m$ and $\x, \y, \widetilde{\x} \in \mathbb{R}^n$,
\begin{equation}
\begin{split}
 \phi_m\left(\mathbf{a}, \mathbf{b}\right) &:= \frac{1}{m} \sum_{i=1}^m \phi\left(a_i, b_i\right), \quad \text{ and } \quad \psi_n\left(\x, \y,\widetilde{\x}\right) := \frac{1}{n} \sum_{i=1}^n \psi\left(x_i, y_i,\widetilde{x}_i\right).
 \label{eq:sum_funcs}
\end{split}
\end{equation}
Then the functions in \eqref{eq:sum_funcs} are uniformly PL(2). Next, assume the following:
\begin{itemize}
\item[\textbf{(A1)}] The measurement matrix $\A$ has i.i.d.\  Gaussian entries with mean $0$ and variance $1/m$.
\item[\textbf{(A2)}] The noise $\w$ 
is i.i.d.\ $\sim f(W)$ with finite  $\mathbb{E}[W^{2}]$.
\item[\textbf{(A3)}] The signal and SI pair $(\x, \widetilde{\x})$ are sampled i.i.d.\ from $f(X, \widetilde{X})$ with finite $\mathbb{E}[X^{2}]$ and finite $ \mathbb{E}[\widetilde{X}^{2}]$. 
\item[\textbf{(A4)}] For $t\geq 0$, the denoisers $\eta_t(\cdot, \cdot)$ defined in \eqref{eq:eta_2_iid} are Lipschitz continuous: for scalars $a_1, a_2, b_1, b_2$, and a constant $L > 0$,
$|\eta_t(a_1, b_1) - \eta_t(a_2, b_2)| \leq L||(a_1, b_1) - (a_2, b_2)||$.
\end{itemize}
Then, we have the following asymptotic results for the functions defined in \eqref{eq:sum_funcs},
\begin{equation}\label{eq:main}
\begin{aligned}
&\lim_m\phi_m\left(\ramp^t, {\w}\right) \overset{p}{=} \lim_m  \mathbb{E}\left[\phi_m\left(\mathbf{W} + \sqrt{\lambda_t^2 - \sigma_w^2} \, {\Z_1}, {\mathbf{W}}\right)\right],\\
&\lim_n \psi_n\left(\x^t + {\A}^T {\ramp}^t,\x, \widetilde{\x}\right)  \overset{p}{=} \lim_n \mathbb{E}\left[\psi_n\left({\mathbf{X}} + \lambda_t {\Z_2}, {\mathbf{X}},\widetilde{\mathbf{X}}\right)\right] ,
\end{aligned}
 \end{equation}
 where 
 $\Z_1\sim\mathcal{N}(\mathbf{0}, \pmb{\mathbb{I}}_m)$, $\Z_2\sim \mathcal{N}(\mathbf{0},\pmb{\mathbb{I}}_n)$
 and both  $\mathbf{Z}_1$ and $\mathbf{Z}_2$ are independent of $\mathbf{W}  \sim i.i.d.\ f(W)$ and $(\mathbf{X}, \widetilde{\mathbf{X}})  \sim i.i.d.\ f(X, \widetilde{X})$.  
The signal estimate $\x^t$ and residual $\ramp^t$
  are defined in the AMP-SI recursion~\eqref{eq:1-5_iid}-\eqref{eq:1-6_iid}, and $\lambda_t$ in the SE~\eqref{eq:SE2_iid}.
\end{theorem}

Section~\ref{main_proof} contains the proof of Theorem~\ref{thm:SE}. The proof follows from Berthier \emph{et al}.~\cite[Theorem 14]{Berthier2017} and the strong law of large numbers (SLLN).  The main details involve showing that assumptions $\textbf{(A1)} - \textbf{(A4)}$ allow us to apply~\cite[Theorem 14]{Berthier2017}.  We also note that by employing different aspects of AMP theory, for example \cite{javanmard2013state} or by generalizing \cite{Rush_Finite18}, we could get almost sure convergence in the asymptotics of Theorem~\ref{thm:SE} given in \eqref{eq:main}; however we instead use the Berthier \emph{et al}.~\cite[Theorem 14]{Berthier2017} result as it extends to the more general non-separable case studied in the next section.

We now show how Theorem~\ref{thm:SE} can be used to relate the SE to the output of the AMP-SI recursion~\eqref{eq:1-5_iid}-\eqref{eq:1-6_iid}, by considering a few interesting pseudo-Lipschitz loss functions. 
\begin{corollary}
\label{cor:SE}
Let assumptions $\textbf{(A1)} - \textbf{(A4)}$ be true.  Then letting $\psi^1: \mathbb{R}^{3} \rightarrow \mathbb{R}$ be the function $\psi^1(x, y, z) = (x-y)^2$, for $\lambda_t^2$ defined in \eqref{eq:SE2_iid}, it follows by Theorem~\ref{thm:SE},
\[\lim_{n \rightarrow \infty}  \frac{1}{n} ||\x^t + \A^T \ramp^t- \x||^2 \overset{p}{=}  \lambda_t^2.\]
Similarly, if $\psi^2: \mathbb{R}^{3} \rightarrow \mathbb{R}$ is defined as $\psi^2(x, y, z) = (\eta_t(x, z)-y)^2$, then by 
Theorem~\ref{thm:SE},
\[\lim_{n \rightarrow \infty}  \frac{1}{n} ||\x^{t+1}- \x||^2 \overset{p}{=}  \delta(\lambda_{t+1}^2 - \sigma_w^2).\]
\end{corollary}
\begin{proof}
It is straightforward to show that $\psi^1$ and $\psi^2$ are both PL(2), and thus Theorem~\ref{thm:SE} can be applied.  Showing $\psi^2$ is PL(2) uses that $\eta_t$ is Lipschitz by assumption $\textbf{(A4)}$.
\end{proof}

\subsection{Non-separable denoiser} \label{sec:non-sep}

We now generalize Theorem~\ref{thm:SE} to the case of general (non-i.i.d.)\,\,dependencies between the signal and SI pair, 
which calls for the AMP algorithm using the non-separable AMP-SI denoiser of~\eqref{eq:eta_2}.

\begin{theorem}
\label{thm:non-separable}
For any sequences of order $k$ uniformly pseudo-Lipschitz functions, $\kappa_m: \mathbb{R}^{2m}\rightarrow \mathbb{R}$ and $\nu_n: \mathbb{R}^{3n} \rightarrow \mathbb{R}$, assume the following.
\begin{itemize}
\item[\textbf{(B1)}] The measurement matrix  $\A$ has i.i.d.\ Gaussian entries with mean $0$ and variance $1/m$.
\item[\textbf{(B2)}] For each $t$, the sequence (in $n$) of denoisers 
$g_{t}(\cdot,\cdot)$ defined in \eqref{eq:eta_2} are uniformly Lipschitz. 
\item[\textbf{(B3)}] 
The signal $\x$ and SI pair $\widetilde{\x}$ are sampled from a joint pdf $f(\X, \widetilde{\X})$ such that, elementwise, there are finite 
fourth moments
$\mathbb{E}[X_i ^4] < \infty$ and $\mathbb{E}[\widetilde{X}_i^4] < \infty$, and equal second moments, $\mathbb{E}[X_i^2] = \mathbb{E}[X_j^2]$ and $\mathbb{E}[\widetilde{X}_i^2] = \mathbb{E}[\widetilde{X}_j^2]$, for all $i, j \in \{1, 2, \ldots, n\}$. Moreover, $\mathrm{Cov}(X_i^2, X_j^2)\to 0$ and $\mathrm{Cov}(\widetilde{X}_i^2,\widetilde{X}_j^2)\to 0$, when $|i-j|\to \infty$.
\item[\textbf{(B4)}] The noise $\w$ 
is i.i.d.\ $\sim f(W)$ with finite second moment, $\mathbb{E}[W^2]$. 
\item[\textbf{(B5)}] For any iterations $s, t\in \mathbb{N} $ and for any $2 \times 2$ covariance matrix $\boldsymbol \Sigma$, the following limits exist,
\begin{align*}
&\lim_{n\to\infty}\frac{1}{n} \mathbb{E}_{\Z}\left[\x^T g_t(\x+\Z,\widetilde{\x})\right]< \infty,\\
&\lim_{n\to\infty}\frac{1}{n} \mathbb{E}_{\Z,\Z'}\left[g_t(\x+\Z,\widetilde{\x})^Tg_s(\x+ \Z',\widetilde{\x})\right] < \infty,
\end{align*}
where $(\Z, \Z')\sim N(\mathbf{0},\boldsymbol \Sigma \otimes \pmb{\mathbb I}_n)$, with $\otimes$ denoting the tensor product,\footnote{The tensor product of matrices $\A \in \mathbb{R}^{m \times n}$ and $\mathbf{B} \in \mathbb{R}^{p \times q}$, denoted $\A \otimes \mathbf{B} \in \mathbb{R}^{mp \times nq}$, equals $\begin{Bmatrix}
\A_{11}\mathbf B & \A_{12}\mathbf B & \dots\\
\A_{21}\mathbf B & \A_{22}\mathbf B & \dots\\ \vdots & \vdots & \ddots
\end{Bmatrix}$.}  and $g_t$, $g_s$ are denoisers defined in~\eqref{eq:eta_2} at iteration $t$ and $s$, respectively. 
\end{itemize}
Then,
\begin{equation}
\begin{aligned}\label{eq:main2}
& \lim_m \kappa_m\left({\ramp}^t, {\w}\right) \overset{p}{=} \lim_m \mathbb{E}_{{\Z_1}}\left[\kappa_m\left({\w} + \sqrt{\lambda_t^2 - \sigma_w^2} \, {\Z_1}, {\w}\right)\right], \\
&\lim_n \nu_n\left({\x}^{t} + {\A}^T {\ramp}^t, {\x},\widetilde{\x}\right) \overset{p}{=} \lim_n \mathbb{E}_{{\Z_2}}\left[\nu_n\left({\x} + \lambda_t {\Z_2}, {\x},\widetilde{\x}\right)\right],
\end{aligned}
\end{equation}
where $\Z_1, \Z_2$ are independent, standard Gaussian vectors. 
The signal estimate $\x^t$ and residual $\ramp^t$ are defined in the AMP-SI recursion~\eqref{eq:1-5}--\eqref{eq:1-6}, and $\lambda_t$ in the SE~\eqref{eq:SE2}.
\end{theorem}

Section~\ref{main_proof_non-separable} contains the proof of Theorem~\ref{thm:non-separable}. Theorem~\ref{thm:non-separable} can be used to relate the SE to the output of the AMP-SI recursion~\eqref{eq:1-5}--\eqref{eq:1-6}, which is easily seen by considering a few specific uniformly pseudo-Lipschitz loss functions.
\begin{corollary} 
\label{cor:SE-nonsep}
Under assumptions $\textbf{(B1)} - \textbf{(B5)}$, letting $\nu_n^1: \mathbb{R}^{3n} \rightarrow \mathbb{R}$ be defined as $\nu_n^1(\x, \y, \z) = \frac{1}{n}||\x-\y||^2$,
then by Theorem~\ref{thm:non-separable},
\[\lim_{n \rightarrow \infty}  \frac{1}{n} ||\x^t + \A^T \ramp^t- \x||^2 \overset{p}{=}  \lambda_t^2,\]
where $\lambda_t^2$ is defined in \eqref{eq:SE2}.  Under an additional assumption on the denoiser $g_t$ defined in \eqref{eq:eta_2}, namely for $\mathbf{0} \in \mathbb{R}^n$, the all-zeros vector, assuming $||g_t(\mathbf{0}, \mathbf{0})||/\sqrt{n} \leq C$ for some constant $C > 0$ that does not depend on $n$,  then similarly considering $\nu^2_n: \mathbb{R}^{3n} \rightarrow \mathbb{R}$ defined as $\nu_n^2(\x, \y, \z) = \frac{1}{n}||g_t(\x, \z)-\y||^2$, by Theorem~\ref{thm:non-separable}, \[\lim_{n \rightarrow \infty} \frac{1}{n} ||\x^{t+1}- \x||^2 \overset{p}{=}\delta(\lambda_{t+1}^2 - \sigma_w^2).\]
\end{corollary}

\begin{proof}
It is straightforward to show that $\nu_n^1$ and $\nu_n^2$ are both uniformly PL(2) functions, and thus Theorem~\ref{thm:non-separable} can be applied.  We show this result in detail for $\nu_n^2(\x, \y, \z) = \frac{1}{n}||g_t(\x,\z)-\y||^2$ in Appendix~\ref{app:technical}, and the result for $\nu_n^1$ follows similarly. 
\end{proof}

\begin{rem}
Note that the assumption $||g_t(\mathbf{0}, \mathbf{0})||/\sqrt{n} \leq C$ with $C$ not depending on $n$ in Corollary~\ref{cor:SE-nonsep} is not a terribly restrictive one.
For example, notice that if $[g_t(\mathbf{0}, \mathbf{0})]_i = [g_t(\mathbf{0}, \mathbf{0})]_j = C$ for all $i, j \in \{1, 2, \ldots, n\}$, with $C$ not growing with $n$, then the assumption would be true and this is a property we could reasonably expect from denoisers.  
Similarly, if $\max_{i \in \{1, 2, \ldots, n\}}[g_t(\mathbf{0}, \mathbf{0})]_i  \leq C$, then the assumption would hold.
\end{rem}

\section{Examples}
\label{examples}
In this section, we consider some example signal and SI models to show how to derive the conditional denoiser in \eqref{eq:eta_2}, use it
to construct the AMP-SI algorithm and SE, 
and then apply the theoretical guarantees of Theorem \ref{thm:SE} or Theorem \ref{thm:non-separable}.

\subsection{Separable Case (Theorem~\ref{thm:SE})} \label{subsec:examples iid}
Before we get to the examples, we state a lemma showing that functions with bounded derivatives are Lipschitz.
\begin{lemma}
\label{lem:Lipschitz}
A function $\phi: \mathbb{R}^2 \rightarrow \mathbb{R}$ having bounded derivatives, $0 < \mathsf{D}_1,  \mathsf{D}_2 < \infty,$
\[\Big \lvert \frac{\partial}{\partial x} \phi(x, y) \Big \lvert \leq \mathsf{D}_1 \qquad \text{and } \qquad \Big \lvert \frac{\partial}{\partial y} \phi(x, y)\Big \lvert \leq \mathsf{D}_2,\]
 is Lipschitz continuous (i.e.\ pseudo-Lipschitz of order 1) with Lipschitz constant $\sqrt{{2}(\mathsf{D}_1^2 + \mathsf{D}_2^2)}$.  
\end{lemma}
\begin{proof}
By the triangle inequality, 
\begin{equation*}
\begin{split}
\lvert \phi(x_1, y_1) - \phi(x_2, y_2) \lvert 
&\leq \lvert \phi(x_1, y_1) - \phi(x_1, y_2) \lvert + \lvert \phi(x_1, y_2) - \phi(x_2, y_2) \lvert  \leq  \mathsf{D}_2 \lvert y_1- y_2 \lvert +  \mathsf{D}_1 \lvert x_1 - x_2 \lvert.
\end{split}
\end{equation*}
Then using the Cauchy-Schwarz property,
\begin{equation*}
\begin{split}
\lvert \phi(x_1, y_1) - \phi(x_2, y_2) \lvert 
&\leq  \sqrt{\mathsf{D}_2^2+  \mathsf{D}_1^2} \,\, \sqrt{ ( y_1- y_2)^2  + (x_1 - x_2)^2}=\sqrt{2(\mathsf{D}_2^2+  \mathsf{D}_1^2)}\,\, \frac{||(x_1, y_1) - (x_2, y_2)||}{\sqrt{2}}.
\end{split}
\end{equation*}
\end{proof}

\subsubsection{Gaussian-Gaussian Signal and SI}
In this model, referred to as the GG model henceforth, the signal has i.i.d.\ Gaussian entries with zero mean and finite variance $\sigma_x^2$, 
and we have access to SI in the form of the signal with additive white Gaussian noise (AWGN). 
In particular, the signal, $\x$, and SI, $\widetilde \x$, are related by
\begin{equation}\label{eq:SI form}
    \widetilde \x= \x+ \mathcal{N}(\mathbf{0}, \sigma^2 \pmb{\mathbb{I}}).
\end{equation}
As was shown in~\cite{Ma2018}, in this case, the AMP-SI denoiser~\eqref{eq:eta_2_iid} equals
\begin{equation}
\begin{aligned}
\eta_{t}(a,b)&=\mathbb{E}\left[X\middle| X+\lambda_{t}Z=a, \widetilde X=b\right] = \frac{\sigma_x^2 \sigma^2 a + \sigma_x^2 \lambda_{t}^2 b}{\sigma_x^2 (\sigma^2 + \lambda_{t}^2) + \sigma^2 \lambda_{t}^2},
\label{eq:GGdenoiser}
\end{aligned}
\end{equation} 
where we highlight that $\eta_t(a,b)$ is linear in $a$ and $b$, because
all related RVs are jointly Gaussian.
Then the SE~\eqref{eq:SE2_iid} can be computed as
\begin{equation}
\begin{split}
\lambda_t^2 = \sigma_w^2 + \frac{1}{\delta}\left[ \frac{ \sigma_x^2 \sigma^2 \lambda_{t-1}^2 }{\sigma_x^2 (\sigma^2 + \lambda_{t-1}^2) + \sigma^2 \lambda_{t-1}^2} \right] .
\end{split}
\end{equation}
We note that  as a result of Lemma \ref{lem:Lipschitz}, because
\[\Big \lvert \frac{\partial}{\partial a} \eta_{t}(a,b) \Big \lvert = \Big \lvert \frac{\sigma_x^2 \sigma^2}{\sigma_x^2 (\sigma^2 + \lambda_{t}^2) + \sigma^2 \lambda_{t}^2} \Big \lvert \leq 1 \quad
\text{ and } \quad \Big \lvert \frac{\partial}{\partial b} \eta_{t}(a,b) \Big \lvert = \Big \lvert \frac{ \sigma_x^2 \lambda_{t}^2}{\sigma_x^2 (\sigma^2 + \lambda_{t}^2) + \sigma^2 \lambda_{t}^2} \Big \lvert \leq 1,\]
the denoiser is pseudo-Lipschitz,
the assumptions $\textbf{(A1)} - \textbf{(A4)}$ are satisfied in the GG case, and we can apply Theorem~\ref{thm:SE}.

\subsubsection{Bernoulli-Gaussian Signal and SI}
The Bernoulli-Gaussian (BG) model reflects a
scenario in which one wishes to recover a sparse signal and has access to SI in the form of the signal with AWGN as in \eqref{eq:SI form}. 
In the BG model, 
each entry of the signal is independently generated according to $x_i\sim \epsilon\mathcal{N}(0,1)+(1-\epsilon)\delta_0$, where $\delta_0$ is the Dirac delta function at $0$. In words, the entries of the signal independently take the value $0$ with probability $1-\epsilon$ and are $\mathcal{N}(0,1)$ with probability $\epsilon$. 
In this case, originally studied in ~\cite{Ma2018}, the AMP-SI denoiser~\eqref{eq:eta_2_iid} equals
\begin{equation}
\begin{aligned}
\eta_{t}(a,b)=\mathbb{E}\left[X\middle| X+\lambda_{t}Z=a, \widetilde X=b\right]&=P \left (X\neq0|a, b \right)\mathbb{E}\left[X\middle|a, b, X\neq 0\right] \\&=P \left (X\neq0|a, b \right)\left( \frac{ \sigma^2 a + \lambda_{t}^2 b}{ \sigma^2 + \lambda_{t}^2 + \sigma^2 \lambda_{t}^2}\right),
\label{eq:BGdenoiser}
\end{aligned}
\end{equation} 
where
\begin{equation}
\begin{split}
P(X\neq0|a, b) &=\left( 1+ T_{a,b}\right)^{-1},
\label{eq:BGprobability}
\end{split}
\end{equation}
and letting $\rho_{\tau^2}(x)$ be the zero-mean Gaussian density with variance $\tau^2$ evaluated at $x$ and defining $\nu_{t} := \sigma^2 \lambda_{t}^2 (\sigma^2+\lambda_{t}^2 + \sigma^2 \lambda_{t}^2)$,
\begin{equation}
\begin{split}
T_{a,b}:=\frac{(1 - \epsilon) \rho_{\lambda_{t}^2}(a) \rho_{\sigma^2}(b) }{ \epsilon\rho_{1 + \sigma^2}(b)   \rho_{\frac{\sigma^2}{1 + \sigma^2} + \lambda_{t}^2}\Big(\frac{b}{1 + \sigma^2}  - a\Big)} &= \left(\frac{1 - \epsilon}{\epsilon}\right) \sqrt{\frac{\sigma^2+\lambda_{t}^2 + \sigma^2 \lambda_{t}^2}{ \lambda_{t}^2 \sigma^2}}\exp\left\{\frac{-( \sigma^2 a + \lambda_{t-1}^2 b)^2}{2 \sigma^2 \lambda_{t-1}^2 (\sigma^2+\lambda_{t}^2 + \sigma^2 \lambda_{t}^2)}\right\}\\
&= \left(\frac{1 - \epsilon}{\epsilon}\right) \left(\frac{\nu_{t}  \sqrt{2 \pi }}{ \lambda_{t}^2 \sigma^2}\right)  \rho_{\nu_{t} }\left( \sigma^2 a + \lambda_{t}^2 b\right).
\label{eq:Tab_def}
\end{split}
\end{equation}
Then the SE in~\eqref{eq:SE2_iid} can be computed as
\begin{equation}
\begin{split}
\lambda_t^2 =\sigma_w^2 + \frac{1}{\delta}\mathbb{E}\left[\left(\frac{ \sigma^2(X + \lambda_{t-1}Z)+\lambda_{t-1}^2\widetilde{X}}{(1 + T_{X + \lambda_{t-1}Z, \widetilde{X}})(\sigma^2 + \lambda_{t-1}^2 + \sigma^2 \lambda_{t-1}^2)} -X\right)^2\right].
\label{eq:BG_SE}
\end{split}
\end{equation}
As before, we use Lemma \ref{lem:Lipschitz} to show that the denoiser defined in \eqref{eq:BGdenoiser} and \eqref{eq:BGprobability} is Lipschitz continuous, and so assumptions $\textbf{(A1)} - \textbf{(A4)}$ are also satisfied in the BG case. Therefore,
we can apply Theorem~\ref{thm:SE}. The technical work showing the bounds
\[
\Big\lvert \frac{\partial}{\partial a} \eta_{t}(a,b) \Big\lvert \leq  1+  \frac{2(1 - \epsilon)}{\sigma_w \epsilon} \quad \text{ and } \quad \Big\lvert \frac{\partial}{\partial a} \eta_{t}(a,b) \Big\lvert \leq  1+  \frac{2(1 - \epsilon)}{\sigma \epsilon},
\]
as needed to apply Lemma \ref{lem:Lipschitz}, can be found in Appendix~\ref{app:examples}.

\subsubsection{Numerical Results} 
We conclude Section~\ref{subsec:examples iid} with numerical results
for scenarios where the signal, $\x$, and SI, $\widetilde{\x}$, are both i.i.d.,
and separable denoisers are appropriate.
We will see that the empirical mean square error (MSE) performance of AMP-SI 
is well-tracked by SE predictions.

\begin{figure}[t]
\centering
\subcaptionbox{n=100}{\includegraphics[width=0.33\textwidth]{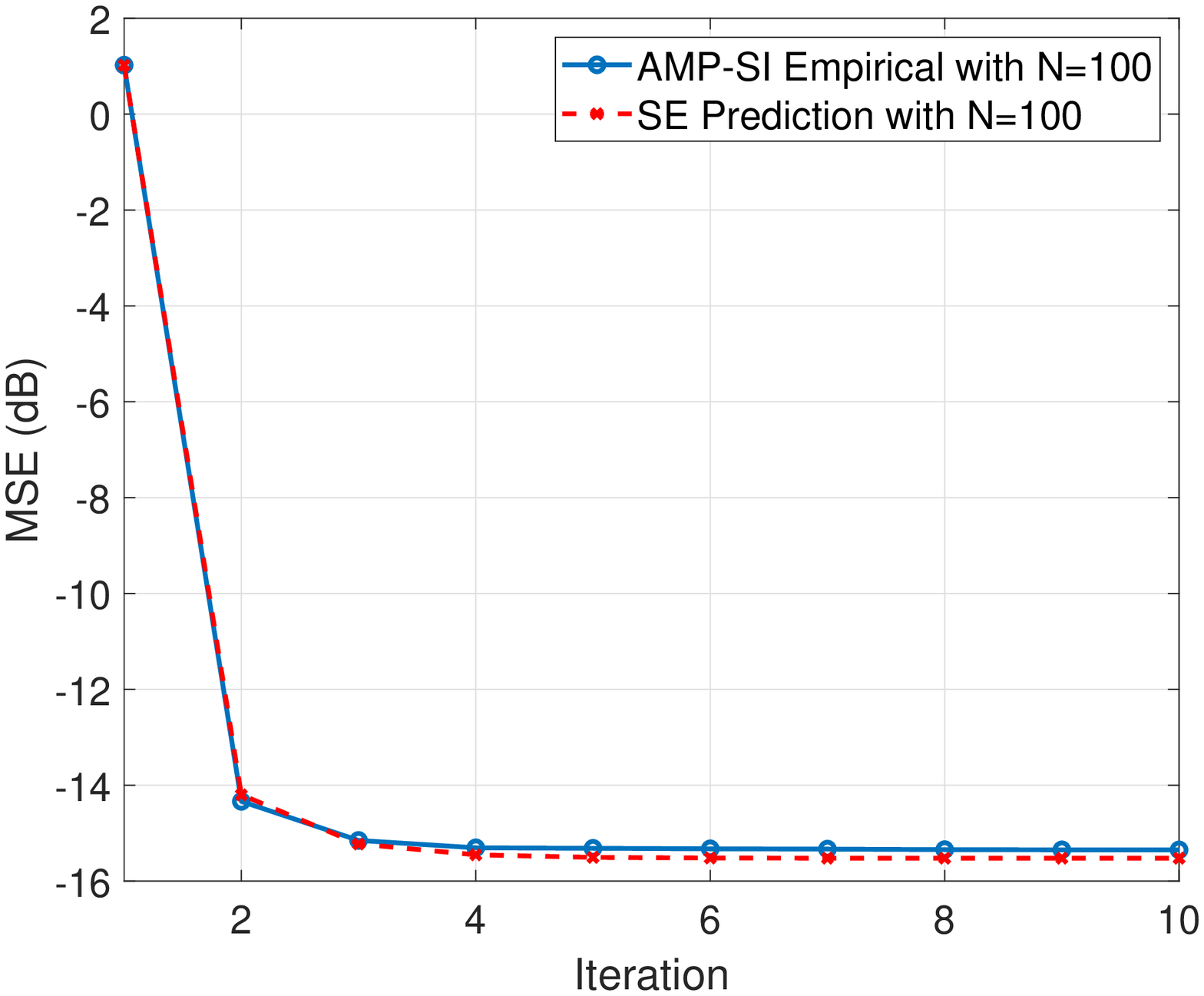}\label{GG_n=100}}%
\hfill
\subcaptionbox{n=1000}{\includegraphics[width=0.33\textwidth]{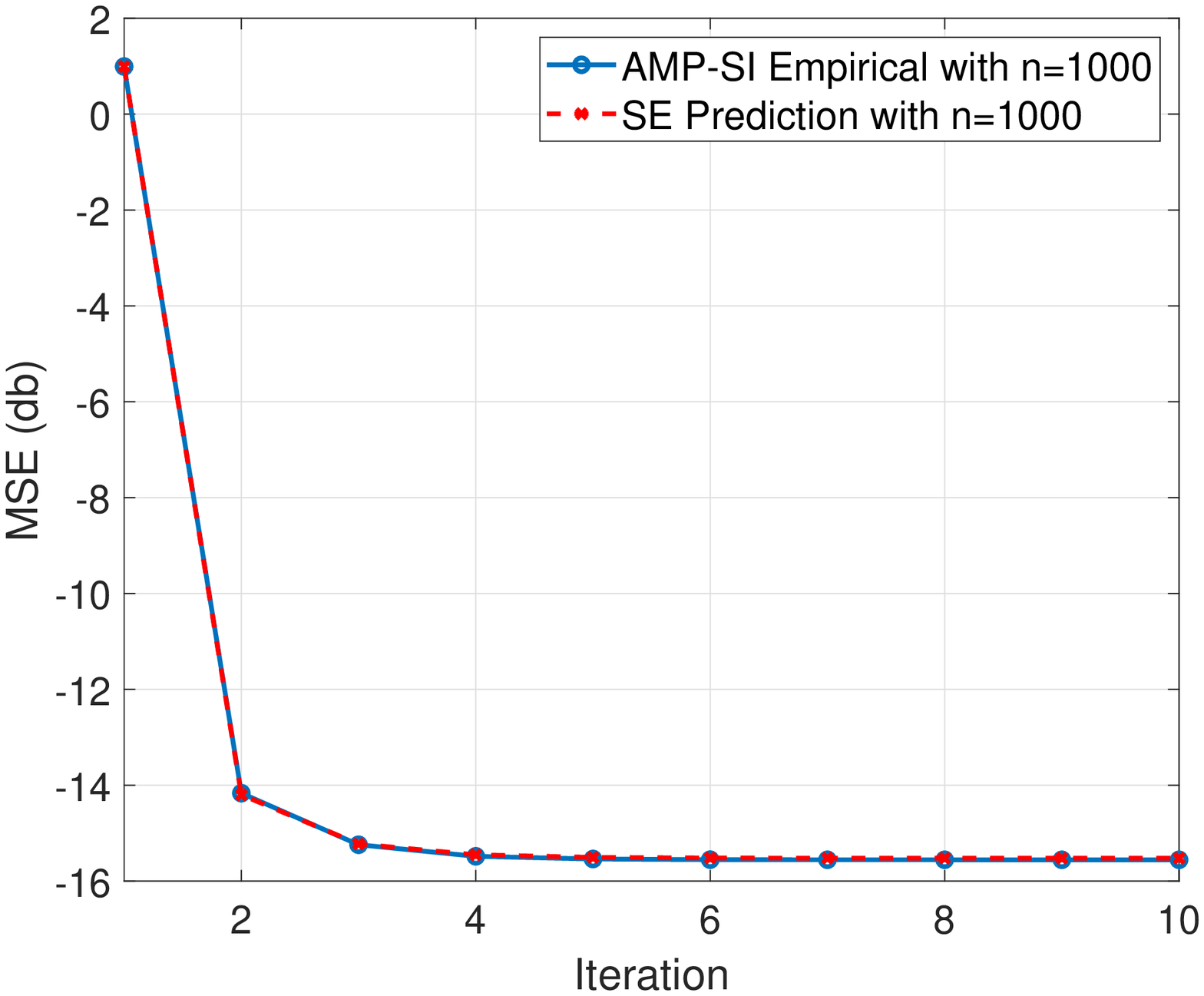}\label{GG_n=1000}}%
\hfill 
\subcaptionbox{n=10000}{\includegraphics[width=0.33\textwidth]{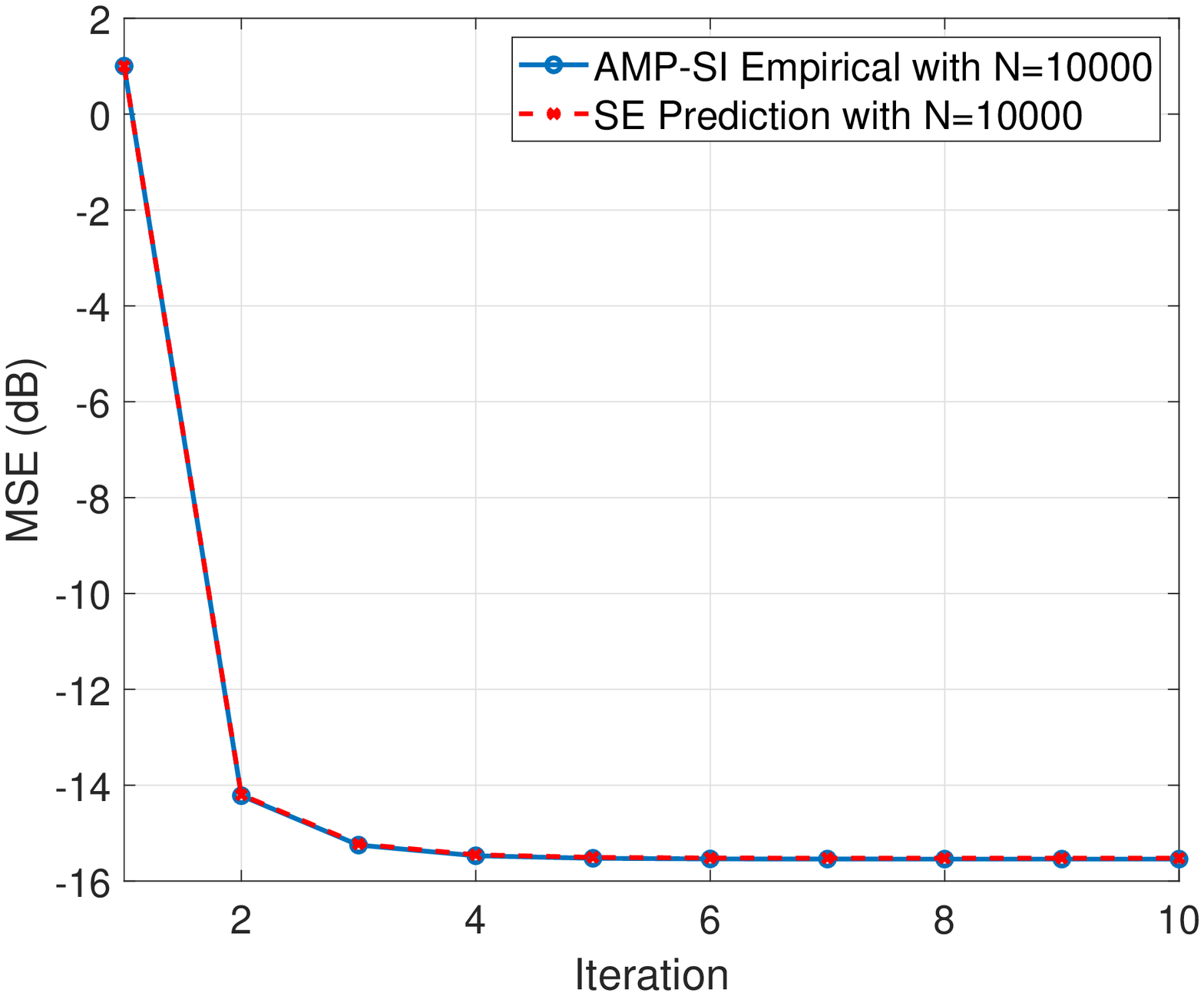}\label{GG_n=10000}}%
\caption{Empirical MSE performance of AMP-SI and SE prediction. (GG model, $\delta=0.3$, $\sigma_{x}=1$, $\sigma_{w}=0.1$, and $\sigma=0.2$.)}
\label{fig:GG}
\end{figure}

Fig.\ \ref{fig:GG} presents results for the GG scenario.
In this example, 
the signal variance $\sigma_x^2=1$, 
the measurement noise variance $\sigma_{w}^2=0.01$, 
and the variance of AWGN in the SI $\sigma^2=0.04$. 
Our empirical MSE results are averaged over 10 trials of GG signal recovery. 
The three panels of the figure contrast different signal lengths. 
For smaller $n$, there is a gap between the empirical MSE and the SE prediction, but
the gap shrinks as $n$ is increased. Overall, the empirical MSE tracks the SE prediction nicely,
especially for larger $n$.

Fig.\ \ref{fig:BG} presents results for the BG scenario.
We again averaged over 10 trials for the empirical results. The signal length $n=10000$, $m=3000$, 
the measurement noise variance $\sigma_{w}^2=0.01$, and $\epsilon=0.2$, where $20\%$ of the 
entries of the signal are nonzero. We consider several variances of AWGN in the SI,
$\sigma^2=0.04$, $\sigma^2=0.25$, and $\sigma^2=1$. 
Again, SE can predict the MSE achieved by AMP-SI at every iteration.

\begin{figure}[t]
 \centering
 \includegraphics[width=0.6\linewidth]{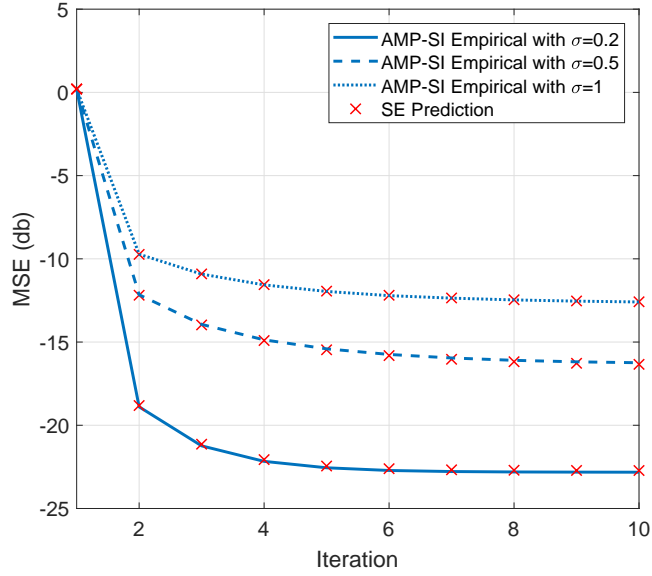}
 \caption{Empirical MSE performance of AMP-SI and SE prediction. (BG model, $n=10000$, $m=3000$, $\epsilon=0.2$, $\sigma_{w}=0.1$.)}
  \label{fig:BG}
\end{figure}

\subsection{Non-Separable Case (Theorem~\ref{thm:non-separable})}
\label{sec:ex:non-sep}

First we provide a 
non-separable extension of
Lemma~\ref{lem:Lipschitz}. 

\begin{lemma}
\label{lem:Lipschitz_non-separable}
A function $f: \mathbb{R}^{2K} \rightarrow \mathbb{R}$ for finite $K$, having bounded partial derivatives, $0 < \mathsf{D}_k,  \mathsf{D}'_k < \infty,$ where $k=1,2,...,K$, 
\[\Big \lvert \frac{\partial}{\partial x_k} f(\x, \y) \Big \lvert \leq \mathsf{D}_k \qquad \text{and } \qquad \Big \lvert \frac{\partial}{\partial y_k} f(\x, \y)\Big \lvert \leq \mathsf{D}'_k,\]
 is Lipschitz continuous (i.e.\ pseudo-Lipschitz of order 1) with Lipschitz constant $\sqrt{2K\sum_{k=1}^{K}(\mathsf{D}_k^2 + \mathsf{D}_k^{'2})}$.
\end{lemma}
\begin{proof}
Using the Triangle Inequality,
\begin{equation*}
\begin{split}
&\lvert f(\x_1, \y_1) - f(\x_2, \y_2)\lvert = \lvert f(\x_1, \y_1) - f(\x_1, \y_2) + f(\x_1, \y_2) - f(\x_2, \y_2) \lvert \\
&\leq \lvert f(\x_1, \y_1) - f(\x_1, \y_2) \lvert + \lvert f(\x_1, \y_2) - f(\x_2, \y_2) \lvert 
\leq \sum_{k=1}^{K} \left( \mathsf{D}'_k \Big\lvert [\y_1]_{k}- [\y_2]_{k} \Big\lvert +  \mathsf{D}_k \Big\lvert [\x_1]_{k} - [\x_2]_{k} \Big\lvert \right).
\end{split}
\end{equation*}
Then, by Cauchy-Schwarz,
\begin{equation*}
\begin{split}
\lvert f(\x_1, \y_1) - f(\x_2, \y_2)\lvert^2 
&\leq \sum_{k'=1}^{K}(\mathsf{D}_{k'}^2 + \mathsf{D}_{k'}^{'2}) \,\, \sum_{k=1}^{K}\Big(( [\y_1]_{k}- [\y_2]_{k})^2  + ([\x_1]_{k}- [\x_2]_{k})^2\Big) \\
&= 2K\sum_{k=1}^{K}(\mathsf{D}_k^2 + \mathsf{D}_k^{'2})\left[ \frac{||{(\x_1, \y_1) - (\x_2, \y_2)||}^2}{2K}\right].
\end{split}
\end{equation*}
\end{proof}

\subsubsection{Block-sparse signal model with AWGN SI}

In this model, we assume that the signal vector $\x$ can be represented by $L$ i.i.d.\ sections or blocks with $K$ entries
per block, where $K$ is finite and fixed, i.e., $n=L\times K$. Among the $K$ elements in a group, 
one is non-zero, taking the value one, and the others are all zeros. We denote a single block $\ell$ of the signal $\x$ as $\x_{(\ell)}$, i.e.\ $\x_{(\ell)} \in \mathbb{R}^K$ contains the $K$ elements in block $\ell$.

Such block-sparse signals are used in sparse superposition codes (SPARCs), which, among other applications, have been studied as codes for both the AWGN channel and the unsourced random access channel. Joseph and Barron~\cite{Joseph2012} introduced SPARCs as an encoding scheme for the AWGN channel and showed that the maximum likelihood decoder is achievable at rates approaching the channel capacity, and then a series of works~\cite{BarbierKrzakala2014ISIT, Rush2017}, 
introduced and rigorously analyzed a computationally-efficient AMP decoder, showing that it can achieve the AWGN capacity asymptotically for SPARCs.  More recently, SPARCs have been studied in the context of the unsourced random access channel~\cite{Giuseppe2018, Rush2020}, where they have also been shown to achieve the channel capacity~\cite{Giuseppe2020}.

We assume that the SI blocks, $\widetilde \x_{(\ell)}$ for $\ell = 1, 2, \ldots, L$, are of the form:
\begin{equation}\label{eq:SI form sparse}
    \widetilde \x_{(\ell)}=\x_{(\ell)}+ \mathcal{N}(\mathbf{0}, \sigma^2 \pmb{\mathbb{I}}_K).
\end{equation}
In this case, the AMP-SI denoiser is separable only across the blocks. 
Within each block $\ell$,
we define a conditional distribution denoiser, $g^{\ell}_t: \mathbb{R}^{2K}\to\mathbb{R}^K$, 
as follows,
\begin{equation}
\label{eq:block_denoiser}
g^{\ell}_t(\mathbf{a}, \mathbf{b}) =\mathbb{E}\left[\X_{(\ell)}\Big|\X_{(\ell)}+\lambda_t\Z_1=\mathbf{a},\widetilde{\X}_{(\ell)}=\mathbf{b}\right] =\mathbb{E}\left[\X_{(\ell)}\Big|\X_{(\ell)}+\lambda_t\Z_1=\mathbf{a},\X_{(\ell)}+\sigma\Z_2=\mathbf{b}\right],
\end{equation}
where $\Z_1\sim \mathcal{N}(\mathbf{0},\pmb{\mathbb{I}}_K)$ is independent of $\Z_2\sim \mathcal{N}(\mathbf{0},\pmb{\mathbb{I}}_K)$.  Then the overall signal denoiser is given by $g_t: \mathbb{R}^{2n}\to\mathbb{R}^n$, defined as
\begin{equation}
\label{eq:overall_denoiser}
g_t(\mathbf{x}, \mathbf{y}) = (g^{1}_t(\mathbf{x}_{(1)}, \mathbf{y}_{(1)}), g^{2}_t(\mathbf{x}_{(2)}, \mathbf{y}_{(2)}), \ldots, g^{L}_t(\mathbf{x}_{(L)}, \mathbf{y}_{(L)})).
\end{equation}
We provide a closed form for the blockwise conditional distribution denoiser given in \eqref{eq:block_denoiser} in the following lemma.

\begin{lemma}
The blockwise denoiser defined in \eqref{eq:block_denoiser} has a closed form as a ratio of exponentials.  For index $i = 1, 2, \ldots, K$,
\begin{equation}
[g^{\ell}_t(\mathbf{a}, \mathbf{b})]_i =\frac{\exp\left\{(a_{i}/\lambda_t^2)+(b_i/\sigma^2)\right\}}{\sum_{k=1}^{K}\exp\left\{({a_k}/{\lambda_t^2})+({b_k}/{\sigma^2})\right\}}.
\label{eq:Lipschitz_den}
\end{equation}
\label{lem:denoiser_simplify}
\end{lemma}

The proof of Lemma \ref{lem:denoiser_simplify} can be found in Appendix~\ref{app:examples} and involves computing the expectation given in \eqref{eq:block_denoiser}.

We now want to prove that we can apply Theorem~\ref{thm:non-separable}, so we will show that assumptions \textbf{(B1)}-\textbf{(B5)} hold in this case.  We first show \textbf{(B2)}, namely that the denoisers in \eqref{eq:overall_denoiser} are uniformly Lipschitz, with the following lemma.

\begin{lemma}
The denoising function $g_t$ defined in \eqref{eq:overall_denoiser} is Lipschitz (i.e.\ pseudo-Lipschitz of order 1) with Lipschitz constant $\sqrt{2} K\Big(\frac{1}{\sigma_w^2} + \frac{1}{\sigma^2}\Big)$.  Therefore $g_t$ is uniformly Lipschitz, since the Lipschitz constant does not depend on $n$.
\label{lem:B2_proof}
\end{lemma}

The proof of Lemma \ref{lem:B2_proof} can be found in Appendix~\ref{app:examples}.

Next, we show that $\textbf{(B5)}$ is satisfied.  Assumption $\textbf{(B5)}$ requires that the following limits exist and are finite:
\begin{equation}
\lim_{n\to\infty}\frac{1}{n} \mathbb{E}_{\Z}[\x^T g_t(\x+\Z,\widetilde{\x})] \qquad \text{ and } \qquad \lim_{n \to\infty}\frac{1}{n} \mathbb{E}_{\Z,\Z'}\left[g_t(\x+ \Z,\widetilde{\x})^T,g_s(\x+ \Z',\widetilde{\x})\right],
\label{eq:blocksparse_B5}
\end{equation}
where $(\Z,\Z') \sim \mathcal{N}(\mathbf{0}, \boldsymbol \Sigma \otimes \pmb{\mathbb{I}}_{n})$.

\begin{lemma}
The limits in \eqref{eq:blocksparse_B5} exist and are finite.
\label{lem:B5_proof}
\end{lemma}

The proof of Lemma \ref{lem:B5_proof} can be found in Appendix~\ref{app:examples}.  It relies on arguing that the SLLN can be applied over the blocks, which are independent, and showing that the expectations that are the limiting
values are finite.
Therefore, $\textbf{(B1)} - \textbf{(B5)}$ are satisfied in the block-sparse signal model and we can apply Theorem~\ref{thm:non-separable}. 

Finally, we will derive the SE~\eqref{eq:SE2} for the block-sparse signal model.  Recall,
\begin{align*}
    \lambda_t^2& = \sigma_w^2 +  \lim_m \frac{1}{m}\mathbb{E}_{\Z}\left[||g_{t-1}(\x + \lambda_{t-1}\Z, \widetilde{\x}) - \x||^2\right] \\
    &=   \sigma_w^2 + \lim_L \frac{1}{\delta L} \sum_{\ell=1}^L \frac{1}{K}\sum_{k=1}^K \mathbb{E}_{\Z}\left[([g^{\ell}_{t-1}(\x_{(\ell)} + \lambda_{t-1}\Z_{(\ell)}, \widetilde{\x}_{(\ell)})]_k - [\x_{(\ell)}]_k)^2\right].
\end{align*}
Plugging in the explicit form of the denoiser given in Lemma~\ref{lem:denoiser_simplify}, we find,
\[\lambda_t^2 = \sigma_w^2 + \frac{1}{\delta} \lim_L \frac{1}{LK}\sum_{\ell=1}^{L}\mathbb{E}_{\Z}
\left[\sum_{k=1}^{K}\left(\frac{\exp\left(([\x_{(\ell)} + \lambda_{t-1}\Z_{(\ell)}]_{k}/\lambda_{t-1}^2)+([\widetilde{\x}_{(\ell)}]_k/\sigma^2)\right)}{\sum_{i=1}^{K}\exp\left(({[\x_{(\ell)} + \lambda_{t-1}\Z_{(\ell)}]_i}/{\lambda_{t-1}^2})+({[\widetilde{\x}_{(\ell)}]_i}/{\sigma^2})\right)}-[\x_{(\ell)}]_k\right)^2\right].\]
Since within any section $\ell = 1, 2, \ldots, L$, each element $k= 1, 2,\ldots, K$ is equally likely to be the non-zero element, we assume without loss of generality (WLOG) that $[\x_{(\ell)}]_1 =1$ and $[\x_{(\ell)}]_2 = [\x_{(\ell)}]_3 = \ldots = [\x_{(\ell)}]_K =0$.  We can then simplify the above,
\begin{align*}
\lambda_t^2 &= \sigma_w^2 + \frac{1}{\delta} \lim_L \frac{1}{LK}\sum_{\ell=1}^{L}\mathbb{E}_{\Z}
\left[\left(\frac{\exp\Big(\frac{1+\lambda_{t-1}[\Z_{(\ell)}]_{1}}{\lambda_{t-1}^2}+\frac{[\widetilde{\x}_{(\ell)}]_1}{\sigma^2}\Big)}{\exp\Big(\frac{1+\lambda_{t-1}[\Z_{(\ell)}]_{1}}{\lambda_{t-1}^2}+\frac{[\widetilde{\x}_{(\ell)}]_1}{\sigma^2}\Big) + \sum_{i=2}^{K}\exp\Big(\frac{ \lambda_{t-1}[\Z_{(\ell)}]_i}{\lambda_{t-1}^2}+\frac{[\widetilde{\x}_{(\ell)}]_i}{\sigma^2}\Big)}-1\right)^2 \right] \\
&\qquad \qquad + \frac{1}{\delta} \lim_L \frac{1}{LK}\sum_{\ell=1}^{L}\mathbb{E}_{\Z}
\left[\sum_{k=2}^{K}\left(\frac{\exp\Big(\frac{\lambda_{t-1}[\Z_{(\ell)}]_{k}}{\lambda_{t-1}^2}+\frac{[\widetilde{\x}_{(\ell)}]_k}{\sigma^2}\Big)}{\exp\Big(\frac{1+\lambda_{t-1}[\Z_{(\ell)}]_{1}}{\lambda_{t-1}^2}+\frac{[\widetilde{\x}_{(\ell)}]_1}{\sigma^2}\Big) + \sum_{i=2}^{K}\exp\Big(\frac{ \lambda_{t-1}[\Z_{(\ell)}]_i}{\lambda_{t-1}^2}+\frac{[\widetilde{\x}_{(\ell)}]_i}{\sigma^2}\Big)}\right)^2\right].
\end{align*}
Our last simplification uses the fact that $\widetilde{\x}_{(\ell)} = \x_{(\ell)}+\widetilde{\z}_{(\ell)} = \x_{(\ell)} + \mathcal{N}(\mathbf{0}, \sigma^2 \pmb{\mathbb{I}}_{K})$. Therefore,
\begin{align*}
\lambda_t^2 &= \sigma_w^2 + \frac{1}{\delta} \lim_L \frac{1}{LK}\sum_{\ell=1}^{L}\mathbb{E}_{\Z}
\left[\left(\frac{\exp\Big(\frac{1+\lambda_{t-1}[\Z_{(\ell)}]_{1}}{\lambda_{t-1}^2}+\frac{1 + [\widetilde{\z}_{(\ell)}]_1}{\sigma^2}\Big)}{\exp\Big(\frac{1+\lambda_{t-1}[\Z_{(\ell)}]_{1}}{\lambda_{t-1}^2}+\frac{1 + [\widetilde{\z}_{(\ell)}]_1}{\sigma^2}\Big) + \sum_{i=2}^{K}\exp\Big(\frac{ \lambda_{t-1}[\Z_{(\ell)}]_i}{\lambda_{t-1}^2}+\frac{[\widetilde{\z}_{(\ell)}]_i}{\sigma^2}\Big)}-1\right)^2 \right] \\
&\qquad \qquad + \frac{1}{\delta} \lim_L \frac{1}{LK}\sum_{\ell=1}^{L}\mathbb{E}_{\Z}
\left[\sum_{k=2}^{K}\left(\frac{\exp\Big(\frac{\lambda_{t-1}[\Z_{(\ell)}]_{k}}{\lambda_t^2}+\frac{[\widetilde{\z}_{(\ell)}]_k}{\sigma^2}\Big)}{\exp\Big(\frac{1+\lambda_{t-1}[\Z_{(\ell)}]_{1}}{\lambda_{t-1}^2}+\frac{1 +[\widetilde{\z}_{(\ell)}]_1}{\sigma^2}\Big) + \sum_{i=2}^{K}\exp\Big(\frac{ \lambda_{t-1}[\Z_{(\ell)}]_i}{\lambda_{t-1}^2}+\frac{[\widetilde{\z}_{(\ell)}]_i}{\sigma^2}\Big)}\right)^2\right].
\end{align*}
From the above, we note that each sum over $\ell$ is a sum of i.i.d.\ RVs (where the randomness is w.r.t. realizations from a RV $\widetilde{\Z}$), and appealing to the   SLLN gives
\begin{align*}
\lambda_t^2 &= \sigma_w^2 + \frac{1}{\delta K}\mathbb{E}_{\Z, \widetilde{\Z}}
\left[\left(\frac{\exp\Big(\frac{1+\lambda_{t-1}[\Z_{(1)}]_{1}}{\lambda_{t-1}^2}+\frac{1 + [\widetilde{\Z}_{(1)}]_1}{\sigma^2}\Big)}{\exp\Big(\frac{1+\lambda_{t-1}[\Z_{(1)}]_{1}}{\lambda_{t-1}^2}+\frac{1 + [\widetilde{\Z}_{(1)}]_1}{\sigma^2}\Big) + \sum_{i=2}^{K}\exp\Big(\frac{ \lambda_{t-1}[\Z_{(1)}]_i}{\lambda_{t-1}^2}+\frac{[\widetilde{\Z}_{(1)}]_i}{\sigma^2}\Big)}-1\right)^2 \right] \\
&\qquad \qquad + \frac{1}{\delta K}\mathbb{E}_{\Z, \widetilde{\Z}}
\left[\sum_{k=2}^{K}\left(\frac{\exp\Big(\frac{\lambda_{t-1}[\Z_{(1)}]_{k}}{\lambda_{t-1}^2}+\frac{[\widetilde{\Z}_{(1)}]_k}{\sigma^2}\Big)}{\exp\Big(\frac{1+\lambda_{t-1}[\Z_{(1)}]_{1}}{\lambda_{t-1}^2}+\frac{1 +[\widetilde{\Z}_{(1)}]_1}{\sigma^2}\Big) + \sum_{i=2}^{K}\exp\Big(\frac{ \lambda_{t-1}[\Z_{(1)}]_i}{\lambda_{t-1}^2}+\frac{[\widetilde{\Z}_{(1)}]_i}{\sigma^2}\Big)}\right)^2\right],
\end{align*}
where we have used the fact that the RVs were i.i.d.\ over $\ell$, and therefore the expectation is the same irrespective of $\ell$ so we have chosen to take the expectation for section $\ell = 1$, denoted on the vectors by $\Z_{(1)}$ and $\widetilde{\Z}_{(1)}$.

\subsubsection{Numerical Results}
We now provide concrete numerical examples where the blockwise denoiser defined in \eqref{eq:Lipschitz_den} is used inside AMP-SI to estimate $\x$. We also contrast these numerical results with those produced by using $K$ different \emph{scalar} denoisers in each section for pair $(\X_{(\ell)},\widetilde \X_{(\ell)})$.  The scalar denoisers will be modeled as pairwise i.i.d., or to be more specific, we treat each entry of $\x$, $x_i$, as Bernoulli with probability $1/K$ and the SI, $\widetilde{\x}$, as $\x$ plus AWGN. In this case, the separable denoiser is
\begin{equation}
\begin{aligned}
\eta_{t}(a,b)=\mathbb{E}\left[X\middle| X+\lambda_{t}Z=a, \widetilde X=b\right]&=P \left (X=1|a, b \right)=\rho_{\lambda_t^2}(a-1)\rho_{\sigma^2}(b-1),
\label{eq:separabledenoiser}
\end{aligned}
\end{equation} 
where we again assume $\rho_{\tau^2}(x)$ to be the zero-mean Gaussian density with variance $\tau^2$ evaluated at $x$.

It can be seen in Fig.\ \ref{fig:blockwise} that the MSE achieved by the blockwise denoiser of \eqref{eq:Lipschitz_den} is smaller than that achieved by the separable denoisers in~\eqref{eq:separabledenoiser}, where we average the empirical MSE results over 10 trials. In addition, the three panels of the figure compare different block lengths $K$. For small $K$, there is a gap between the MSE achieved by the blockwise denoiser and the separable denoisers; the gap increases as $K$ is increased.
\begin{figure}[h]
\centering
\subcaptionbox{K=5}{\includegraphics[width=0.33\textwidth]{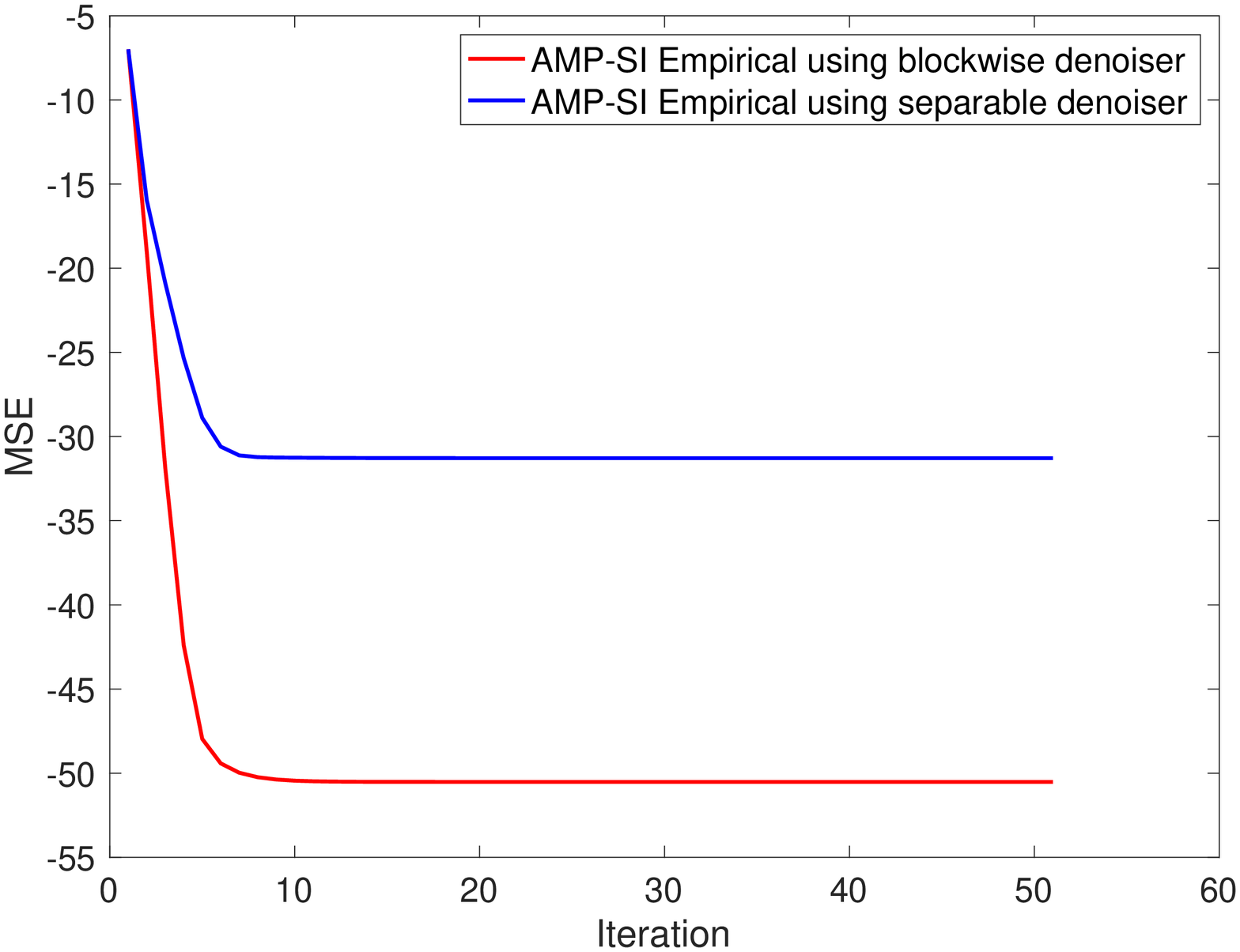}\label{K=5}}%
\hfill
\subcaptionbox{K=10}{\includegraphics[width=0.33\textwidth]{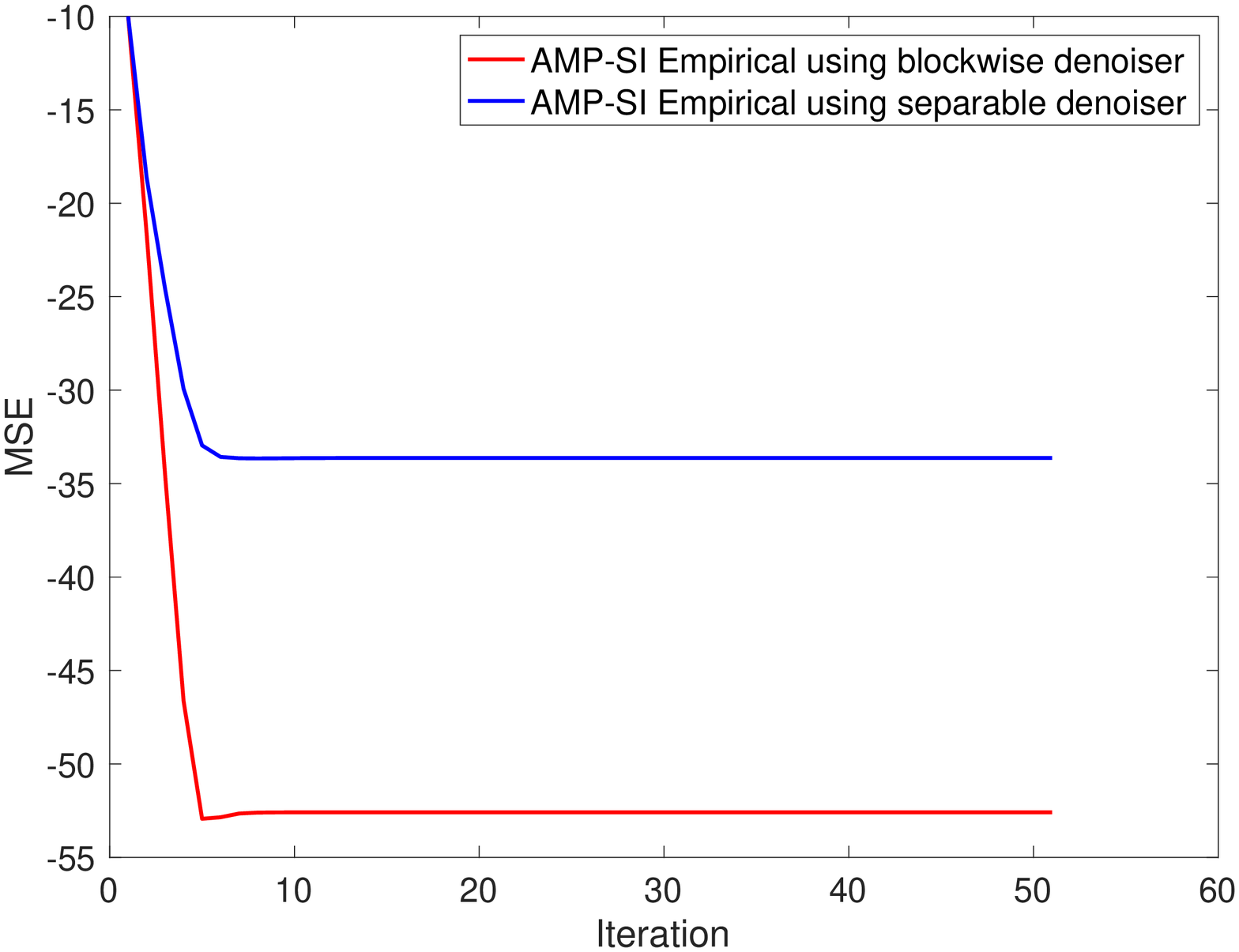}\label{K=10}}%
\hfill 
\subcaptionbox{K=20}{\includegraphics[width=0.33\textwidth]{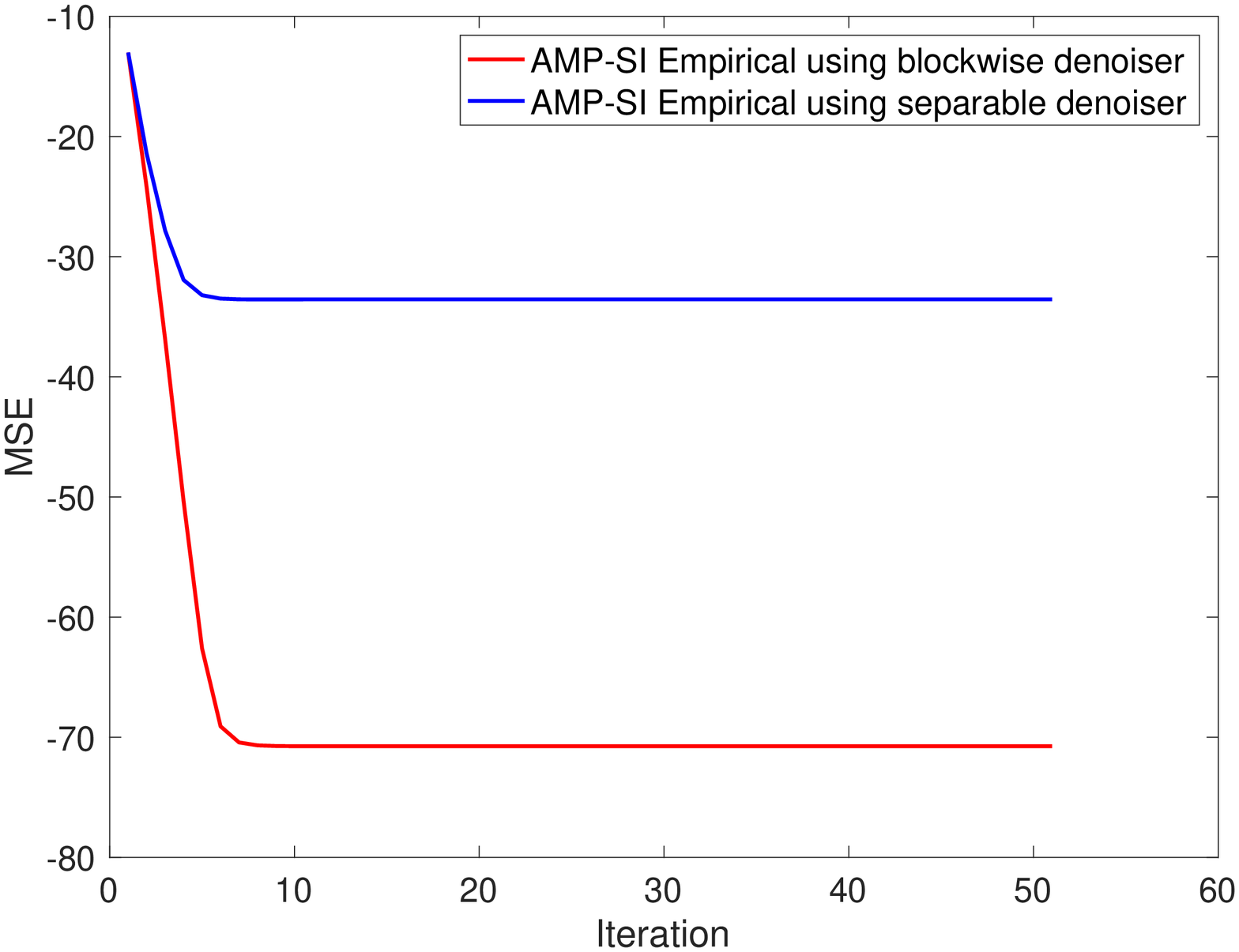}\label{K=20}}%
\caption{Empirical MSE performance of AMP-SI with blockwise denoisers and separable denoisers. ($n=8000$, $\delta=0.3$, $\sigma_{w}^2=0.04$, and $\sigma^2=0.08$.)}
\label{fig:blockwise}
\end{figure}

\section{Proof of Theorem~\ref{thm:SE}}\label{main_proof}

Recall that in this case we assume that $(\x, \widetilde{\x})$ are sampled i.i.d.\ 
from the joint pdf $f(X, \widetilde{X})$, and the conditional denoiser of AMP-SI is given in \eqref{eq:eta_2_iid} for the AMP algorithm in \eqref{eq:1-5_iid}-\eqref{eq:1-6_iid}.  In this case, the simplified SE is given in \eqref{eq:SE2_iid}.

The proof proceeds in three steps. First we show that the functions defined in \eqref{eq:sum_funcs} are uniformly PL(2) when $\phi$ and $\psi$ are PL(2).  This is a straightforward application of Cauchy-Schwarz.  

In the second and third steps, the aim is show that the asymptotic results given in \eqref{eq:main} are true.
In the second step, we show that our assumptions \textbf{(A1)}-\textbf{(A4)} will allow us to make an appeal to Berthier \emph{et al}.~\cite[Theorem 14]{Berthier2017}, which provides a relationship between a general SE and the AMP algorithm when the denoiser is non-separable. 
Finally, in the third step we apply the Berthier \emph{et al}.~\cite[Theorem 14]{Berthier2017} result and argue that the SLLN allows us to include the SI.

\subsection{Step 1}
In step $1$, our goal is to show that the functions defined in \eqref{eq:sum_funcs} are uniformly PL(2) when $\phi$ and $\psi$ are PL(2).
We show the result for $\phi$, and the result for $\psi$ follows similarly.

First, by the fact that $\phi$ is PL(2) ,
\begin{equation*}
\begin{aligned}
\lvert\phi_m(\mathbf{a}, \widetilde{\mathbf{a}}) - \phi_m(\mathbf{b}, \widetilde{\mathbf{b}}) \lvert 
&\leq \frac{1}{m}\sum_{i=1}^m \lvert\phi(a_i, \widetilde{a}_i)- \phi(b_i, \widetilde{b}_i) \lvert \\
&\leq \frac{L}{m}\sum_{i=1}^m\Big[1 +  \frac{\lvert \lvert(a_i, \widetilde{a}_i) \lvert \lvert}{\sqrt{2}}   +  \frac{\lvert \lvert(b_i, \widetilde{b}_i) \lvert \lvert }{\sqrt{2}} \Big]
 \frac{\lvert \lvert(a_i, \widetilde{a}_i) - (b_i, \widetilde{b}_i) \lvert \lvert}{\sqrt{2}}.
\end{aligned}
\end{equation*}
Next we apply Cauchy-Schwarz:
\begin{equation*}
\begin{aligned}
\lvert\phi_m(\mathbf{a}, \widetilde{\mathbf{a}}) - \phi_m(\mathbf{b}, \widetilde{\mathbf{b}}) \lvert^2 & \leq  \frac{L^2}{m}\sum_{i=1}^m\Big[1 +  \frac{\lvert \lvert(a_i, \widetilde{a}_i) \lvert \lvert}{\sqrt{2}}   +  \frac{\lvert \lvert(b_i, \widetilde{b}_i) \lvert \lvert}{\sqrt{2}} \Big]^2\frac{\lvert\lvert(\mathbf{a}, \widetilde{\mathbf{a}}) - (\mathbf{b}, \widetilde{\mathbf{b}}) \lvert\lvert^2}{2m}
\\& \leq 3L^2\Big[1 +  \frac{\lvert \lvert(\mathbf{a}, \widetilde{\mathbf{a}}) \lvert \lvert^{2}}{{2m}}   +  \frac{\lvert \lvert(\mathbf{b}, \widetilde{\mathbf{b}}) \lvert \lvert^{2} }{{2m}} \Big]
\frac{\lvert\lvert(\mathbf{a}, \widetilde{\mathbf{a}}) - (\mathbf{b}, \widetilde{\mathbf{b}}) \lvert\lvert^2}{2m}.
\end{aligned}
\end{equation*}
In the final inequality in the above, we have used  Cauchy-Schwarz in the following way: for any $r >0$ and $a_1, a_2, \ldots, a_m$ scalars, $(|a_1| +|a_2| + \ldots |a_m|)^r \leq m^{r-1}(|a_1|^r +|a_2|^r + \ldots |a_m|^r)$. Namely,
\begin{align*}
\frac{1}{m} \sum_{i=1}^m\Big[1 +  \frac{\lvert \lvert(a_i, \widetilde{a}_i) \lvert \lvert}{\sqrt{2}}   +  \frac{\lvert \lvert(b_i, \widetilde{b}_i) \lvert \lvert}{\sqrt{2}} \Big]^2 &\leq  3\sum_{i=1}^m\Big[\frac{1}{m} + \frac{\lvert \lvert(a_i, \widetilde{a}_i) \lvert \lvert^{2}}{2m}   +  \frac{\lvert \lvert(b_i, \widetilde{b}_i) \lvert \lvert^{2} }{2m} \Big]
\\
&=  3\Big( 1 + \sum_{i=1}^m \frac{a_i^2 + \widetilde{a}_i^2}{2m} + \sum_{i=1}^m \frac{b_i^2 + \widetilde{b}_i^2 }{2m}\Big).
\end{align*}
Finally, we note that this implies that the function $\phi_m$ is uniformly PL(2) as well.  Namely, we have the upper bound,
\begin{equation*}
\begin{aligned}
&\lvert\phi_m(\mathbf{a}, \widetilde{\mathbf{a}}) - \phi_m(\mathbf{b}, \widetilde{\mathbf{b}}) \lvert \leq \sqrt{3}L\Big[ 1 +  \frac{||(\mathbf{a}, \widetilde{\mathbf{a}})||}{\sqrt{2m}}   + \frac{||(\mathbf{b}, \widetilde{\mathbf{b}})||}{\sqrt{2m}}  \Big] \frac{ \lvert \lvert(\mathbf{a}, \widetilde{\mathbf{a}}) - (\mathbf{b}, \widetilde{\mathbf{b}}) \lvert \lvert}{\sqrt{2m}}.
\end{aligned}
\end{equation*}

\subsection{Step 2}
\label{sec:step2}
In this step, we show that our assumptions \textbf{(A1)}-\textbf{(A4)} will allow us to use Berthier \emph{et al}.~\cite[Theorem 14]{Berthier2017} to relate the SE equations to the AMP algorithm.  We first restate ~\cite[Theorem 14]{Berthier2017} for convenience, as it relates to the general AMP algorithm in \eqref{eq:1-5}-\eqref{eq:1-6} and SE \eqref{eq:SE2}.  Then we use the fact that the AMP algorithm~\eqref{eq:1-5_iid}-\eqref{eq:1-6_iid} and SE~\eqref{eq:SE2_iid} studied by Theorem~\ref{thm:SE}, which uses the denoiser $\eta_t$ defined in \eqref{eq:eta_2_iid}, is a special case of the more general setting.
First we note that to apply~\cite[Theorem 14]{Berthier2017}, it is enough if one's problem satisfies the following assumptions:
\begin{itemize}
\item[\textbf{(C1)}] The measurement matrix $\A$ has Gaussian entries with i.i.d.\  mean $0$ and variance $1/m$.
\item[\textbf{(C2)}] 
Define a sequence of denoisers $\widetilde{\eta}_{n}^t: \mathbb{R}^n \rightarrow \mathbb{R}^n$ to be those that apply the denoiser $g_t$ defined in \eqref{eq:eta_2} as follows: $\widetilde{\eta}_{n}^t({\x}) := g_t({\x}, \widetilde{{\x}})$.  For each $t$, the sequence (in $n$) of denoisers $\widetilde{\eta}_{n}^t(\cdot)$ is uniformly Lipschitz.
\item[\textbf{(C3)}]  $||\x||_2^2/n$ converges to a constant as $n\to\infty$.
\item[\textbf{(C4)}] The limit $\sigma_w=\lim_{m\to\infty}{||\w||_2}/{\sqrt{m}}$ is finite.
\item[\textbf{(C5)}]
For any iterations $s, t\in \mathbb{N} $ and for any $2 \times 2$ covariance matrix $\boldsymbol \Sigma$, the following limits exist,
\begin{align*}
&\lim_{n\to\infty}\frac{1}{n} \mathbb{E}_{\Z}[\x^T \widetilde{\eta}_{n}^t(\x+\Z)]< \infty,\\
&\lim_{n\to\infty}\frac{1}{n} \mathbb{E}_{\Z,\Z'}\left[\widetilde{\eta}_{n}^t(\x+ \Z)^T
\widetilde{\eta}_{n}^s(\x+ \Z')\right] < \infty,
\end{align*}
where $(\mathbf{Z}, \mathbf{Z}')\sim N(\mathbf{0}, \boldsymbol \Sigma \otimes \pmb{\mathbb I}_n)$, with $\otimes$ denoting the tensor product.
\end{itemize}
\begin{theorem}{\cite[Theorem 14]{Berthier2017}}
\label{thm:Berthier}
Under the assumptions $\textbf{(C1)} - \textbf{(C5)}$, for any sequences of uniformly pseudo-Lipschitz functions 
$\rho_m: \mathbb{R}^{m} \times \mathbb{R}^{m} \rightarrow \mathbb{R}$ and $\gamma_n: \mathbb{R}^{n} \times \mathbb{R}^{n} \rightarrow \mathbb{R}$,
\begin{equation*}
\begin{aligned}
&\lim_m ( \rho_m({\ramp}^t, {\w}) -  \mathbb{E}_{{\Z_1}}[\rho_m({\w} + \sqrt{\lambda_t^2 - \sigma_w^2} \, {\Z_1}, {\w})])\overset{p}{=}0, \\
&\lim_n\left( \gamma_n({\x}^{t} + {\A}^T {\ramp}^t, {\x}) - \mathbb{E}_{{\Z_2}}\left[\gamma_n({\x} + \lambda_t {\Z_2}, {\x})\right] \right) \overset{p}{=}0,
\end{aligned}
\end{equation*}
where $\Z_1\sim \mathcal{N}(\mathbf{0},\pmb{\mathbb{I}}_m)$, $\Z_2\sim \mathcal{N}(\mathbf{0},\pmb{\mathbb{I}}_n)$, $\x^t$ and $\ramp^t$ are defined in the AMP-SI recursion\eqref{eq:1-5}-\eqref{eq:1-6}, and $\lambda_t$ is defined in the SE~\eqref{eq:SE2}.
\end{theorem}

We now want to apply Theorem~\ref{thm:Berthier}, but have it relate to the AMP algorithm in~\eqref{eq:1-5_iid}-\eqref{eq:1-6_iid} and SE in~\eqref{eq:SE2_iid} studied by Theorem~\ref{thm:SE}, which uses the denoiser $\eta_t$ defined in \eqref{eq:eta_2_iid}.  To do this, we note that the sequence of denoisers $\widetilde{\eta}_{n}^t(\cdot)$ used in assumptions \textbf{(C2)} and \textbf{(C5)} will be those that apply the denoiser $\eta_t$ in \eqref{eq:eta_2_iid} entrywise to its vector inputs.  Specifically,
\begin{equation}
\widetilde{\eta}_{n}^t({\x}) := \eta_t({\x}, \widetilde{{\x}}).
\label{eq:tilde_eta1}
\end{equation}

Now we would like to show that assumptions \textbf{(A1)}-\textbf{(A4)} demonstrate \textbf{(C1)}-\textbf{(C5)} (for $\widetilde{\eta}_{n}^t$ defined in \eqref{eq:tilde_eta1}) in order to apply Theorem~\ref{thm:Berthier}.  First we remind the reader of the strong law, a tool that we will use throughout.
\begin{theorem}
\label{thm:LLN}
{\textbf{\emph{Strong Law of Large Numbers (SLLN)}}}~\cite{JR2006}: Let
$X_1, X_2,...$ be a sequence of i.i.d.\ RVs with finite mean
$\mu$. Then 
\begin{equation}\label{eq:2-1}
P \left(\lim_{n\to\infty}\frac{1}{n}(X_1+X_2+...+X_n)=\mu\right)=1.
\end{equation}
In words, the partial averages $\frac{1}{n}(X_1+X_2+...+X_n)$ converge almost surely
to $\mu < \infty$.
\end{theorem}

Now we demonstrate that our assumptions $\textbf{(A1)} - \textbf{(A4)}$ stated in Section~\ref{main_result} are enough to satisfy the assumptions $\textbf{(C1)} - \textbf{(C5)}$ needed to apply Theorem~\ref{thm:Berthier}. 

Assumptions $\textbf{(A1)}$ and $\textbf{(C1)}$ are identical.  In what follows, we will show that $\textbf{(C2)}$ follows from $\textbf{(A4)}$, $\textbf{(C4)}$ follows from $\textbf{(A2)}$, and $\textbf{(C3)}$ follows from $\textbf{(A3)}$.  Finally we show that $\textbf{(C5)}$ follows from $\textbf{(A3)}$ and $\textbf{(A4)}$.

First consider assumption $\textbf{(C2)}$. The non-separable denoiser $\widetilde{\eta}_{n}^t(\x) = \eta_t({\x}, \widetilde{\x})$ applies the AMP-SI denoiser defined in \eqref{eq:eta_2} entrywise to its vector inputs. From $\textbf{(A4)}$, $\{\eta_t(\cdot, \cdot)\}_{t \geq 0}$ are Lipschitz continuous. Thus, for length-$n$ vectors $\x_1, \x_2$, and fixed SI $\widetilde{\x}$,
\begin{equation*}
\begin{split}
||\widetilde{\eta}_{n}^t({\x_1}) - \widetilde{\eta}_{n}^t({\x_2})||^2 &= \sum_{i=1}^n (\eta_t([\x_1]_i, \widetilde{x}_i) - \eta_t([\x_2]_i, \widetilde{x}_i))^2 \\
&\leq  \sum_{i=1}^n \frac{L^2}{2} ([\x_1]_i - [\x_2]_i)^2 =  \frac{L^2}{2} ||\x_1 - \x_2||^2,
\end{split}
\end{equation*}
and so 
$||\widetilde{\eta}_{n}^t({\x_1}) - \widetilde{\eta}_{n}^t({\x_2})||/\sqrt{n} \leq  L||\x_1 - \x_2||/\sqrt{2n}.$
The Lipschitz constant does not depend on $n$, so $\widetilde{\eta}_{n}^t({\cdot})$ is uniformly Lipschitz.

Now consider assumption $\textbf{(C4)}$. From $\textbf{(A2)}$, the measurement noise $\w$ in~\eqref{eq:1-1} has i.i.d.\ entries $\sim f(W)$ with zero-mean and finite $\mathbb{E}[W^2]$ for $W \sim f(W)$. Then applying the SLLN (Definition~\ref{thm:LLN}), 
\begin{equation*}
\lim_{m\to\infty}\frac{1}{m}||\w||_2^2=\lim_{m\to\infty}\frac{1}{m}\sum_{i=1}^mw_i^2
=\sigma_w^2 = \mathbb{E}[W^2] < \infty.
\end{equation*}
The proof of \textbf{(C3)} follows similarly using  the SLLN and the finiteness of $\mathbb E[X^2]$ assumed in \textbf{(A3)}.

We now show that $\textbf{(C5)}$ is met. Recall $Z\sim \mathcal{N}( \mathbf{0},\sigma_z^2 \pmb{\mathbb I}_n)$. Define $y_i :=x_i\mathbb{E}_{Z} \left[ \eta_t(x_i + Z_{i}, \widetilde{x}_i)\right]$ for $i = 1, 2, \ldots, n$. By assumption $\textbf{(A3)}$, the signal and SI $(\x, \widetilde{\x})$ are sampled i.i.d.\ from the joint density $f(X, \widetilde{X})$. It follows that $y_1, y_2, \ldots, y_n$ are also i.i.d., so by  Theorem~\ref{thm:LLN} if $\mathbb{E}[X\eta_t(X + Z,\widetilde X)] < \infty$ where $Z \sim \mathcal{N}(0,\sigma_z^2)$ is independent of $(X, \widetilde{X}) \sim f(X, \widetilde{X})$, then
\begin{equation*}
\begin{aligned}
\lim_{n\to\infty}\frac{1}{n}\sum_{i=1}^n x_i\mathbb{E}_{Z} \left[\eta_t(x_i+Z_{i},\widetilde x_i)\right] =  \mathbb{E}[X\eta_t(X + Z,\widetilde X)].
\end{aligned}
\end{equation*}
We now show that $\mathbb{E}[X\eta_t(X + Z,\widetilde X)] < \infty$.
First note that $\textbf{(A4)}$ assumes that $\eta_t(\cdot, \cdot)$ is Lipschitz, 
and therefore by Lemma~\ref{prop_Lipschitz},
for constant $L' = \max\{2L, |\eta_t(0, 0)|\}> 0$,
\begin{equation}
  |\eta_t(a_1, b_1)|  \leq L'\Big(1 + \frac{1}{\sqrt{2}}||(a_1, b_1)||\Big)  \leq L'(1 + |a_1| + |b_1|).
\label{eq:eta_bound_1}
\end{equation}
Now using \eqref{eq:eta_bound_1} and the triangle inequality,
\begin{equation}
\begin{aligned}
 \mathbb{E} [X \eta_t(X + Z, \widetilde{X})] &\leq L'\mathbb{E} [|X|  (1+ |X+Z| + |\widetilde{X}|)]\\
&\leq L'(  \mathbb{E} |X|  + \mathbb{E} [X^2] + \mathbb{E} |X|  \mathbb{E}|Z| + \mathbb{E} |X \widetilde{X}|) .
\label{eq:C5_bound2}
\end{aligned}
\end{equation}
Finally, by assumption $\textbf{(A3)}$ we have that $\mathbb{E}[X^2]$ and $\mathbb{E}[\widetilde{X}^2]$ are finite. Moreover, $\mathbb{E} [|X\widetilde{X}|]$ is also finite since $\mathbb{E} [|X\widetilde{X}|] \leq (\mathbb{E} [X^2])^{1/2} (\mathbb{E} [\widetilde{X}^2])^{1/2}$ by  Hölder's inequality.  Noting that for any random variable, $Y$, we have  $|Y|^r \leq 1 + |Y|^k$ for $1 \leq r \leq k$, meaning $\mathbb{E}[|Y|^r] < 1+ \mathbb{E}[|Y|^k]$, the boundedness of $ \mathbb{E} [X \eta_t(X + Z, \widetilde{X})]$ follows from \eqref{eq:C5_bound2} with assumption $\textbf{(A3)}$.

The proof of the second equation in $\textbf{(C5)}$ follows similarly to the proof of the first equation in $\textbf{(C5)}$.  Recall $(Z, Z')\sim N(\mathbf{0}, \boldsymbol \Sigma \otimes \pmb{\mathbb I}_n)$. Define $y_i :=\mathbb{E}_{Z,Z'}[\eta_t(x_i+ Z_i,\widetilde{x}_i)\eta_s(x_i+ Z'_i,\widetilde{x}_i)] $ for $i = 1, 2, \ldots, n$.  By assumption $\textbf{(A3)}$, the signal and SI $(\X, \widetilde{\X})$ are sampled i.i.d.\ from the joint density $f(X, \widetilde{X})$. It follows that $y_1, y_2, \ldots, y_n$ are also i.i.d., so by Theorem~\ref{thm:LLN} if $\mathbb{E} [\eta_t(X+Z,\widetilde X)\eta_s(X+Z',\widetilde X)] < \infty$ where $Z \sim \mathcal{N}(0,\sigma_z^2)$ and $Z' \sim \mathcal{N}(0,\sigma_{z'}^2)$, independent of $(X, \widetilde{X}) \sim f(X, \widetilde{X})$, then
\begin{equation*}
\begin{aligned}
\lim_{n\to\infty}\frac{1}{n} \sum_{i=1}^n \mathbb{E}_{Z,Z'}[\eta_t(x_i+ Z_i,\widetilde{x}_i)\eta_s(x_i+ Z'_i,\widetilde{x}_i)]= \mathbb{E}[\eta_t(X+Z,\widetilde X)\eta_s(X+Z',\widetilde X)].
\end{aligned}
\end{equation*}
We will now show that $\mathbb{E}[\eta_t(X+Z,\widetilde X)\eta_s(X+Z',\widetilde X)] < \infty$.
Using the bound \eqref{eq:eta_bound_1},
\begin{equation*}
\begin{aligned}
\mathbb{E}[\eta_t(X+Z,\widetilde X)\eta_s(X+Z',\widetilde X)] & \leq \mathbb{E}[|\eta_t(X+Z,\widetilde X)| |\eta_s(X+Z',\widetilde X)|]  \\
&\leq {L'}^2\mathbb{E} [(1+ |X+Z| + |\widetilde{X}|)  (1+ |X+Z'| + |\widetilde{X}|)].
\end{aligned}
\end{equation*}
Then using the triangle inequality,
\begin{equation*}
\begin{aligned}
&\mathbb{E} \left[(1+ |X+Z| + |\widetilde{X}|)  (1+ |X+Z'| + |\widetilde{X}|)\right] \leq \mathbb{E} \left[(1+ |X| + |Z| + |\widetilde{X}|)  (1+ |X| + |Z'| + |\widetilde{X}|)\right] \\
&=   1+  2\mathbb{E} \left[|X|\right]  + 2 \mathbb{E} [|\widetilde{X}|]  + 2\mathbb{E} [|X| |\widetilde{X}|]  + \mathbb{E} [ X^2] + \mathbb{E} [\widetilde{X}^2] \\
& \qquad+ \mathbb{E} [|X||Z'|]  + \mathbb{E} [|X||Z|] +\mathbb{E} [ |\widetilde{X}||Z'| ] + \mathbb{E} [ |\widetilde{X}| |Z|  ] +  \mathbb{E} |Z|   + \mathbb{E} [|Z'|]   + \mathbb{E} \left[|Z| |Z'|\right]  \\
&=  1+  2\mathbb{E} [|X|]  + 2 \mathbb{E} [|\widetilde{X}|]  + 2\mathbb{E} [|X \widetilde{X}|]  + \mathbb{E} [ X^2] + \mathbb{E} [\widetilde{X}^2] \\
&  \qquad+  \mathbb{E} [|Z'|](1 + \mathbb{E} [|X|]+ \mathbb{E}  [|\widetilde{X}|] )  +  \mathbb{E} [|Z|] ( 1 + \mathbb{E} [|X|] + \mathbb{E}  [|\widetilde{X}|] )    + \mathbb{E} \left[|ZZ'|\right] .
\end{aligned}
\end{equation*}
The above is finite.  This follows since $\mathbb{E} [|Z'|]$ and $\mathbb{E} [|Z|]$ are finite, as $(Z, Z')$ are Gaussian RVs with finite variance.  Moreover, $\mathbb{E} [|ZZ'|]$ is finite since according to Hölder's inequality, $\mathbb{E} [|ZZ'|]\leq (\mathbb{E} [Z^2])^{1/2} (\mathbb{E} [Z'^2])^{1/2}$. Similarly, all expectations involving $|X|, |\widetilde{X}|$ or their products or squares are finite by assumption \textbf{(A3)}.

We have therefore shown that \textbf{(C1)}-\textbf{(C5)} follows from assumptions \textbf{(A1)}-\textbf{(A4)}.

\subsection{Step 3}
\label{Step 3}
Now that we have justified $\textbf{(C1)} - \textbf{(C5)}$, we make an appeal to Theorem~\ref{thm:Berthier} and the SLLN in order to finally prove \eqref{eq:main}.
The first result in \eqref{eq:main}, namely the asymptotic result for $\phi_m$ uniformly PL(2), follows \emph{almost} immediately by applying Theorem~\ref{thm:Berthier} using $ \rho_m = \phi_m$. Namely, by Theorem~\ref{thm:Berthier}, 
$$\lim_m ( \phi_m({\ramp}^t, {\w}) -  \mathbb{E}_{{\Z_1}}[\phi_m({\w} + \sqrt{\lambda_t^2 - \sigma_w^2} \, {\Z_1}, {\w})])\overset{p}{=}0$$
since $\phi_m$ is to be uniformly PL(2) as was shown in Step 1 above. Note that in the above, $\Z_1 \sim \mathcal{N}(\mathbf{0}, \pmb{\mathbb{I}}_m)$ is independent of $\w$.

To complete the proof, we will finally prove that
\begin{equation}
\begin{split}
\lim_m  \mathbb{E}_{{\Z_1}}[\phi_m({\w} + \sqrt{\lambda_t^2 - \sigma_w^2} \, {\Z_1}, {\w})]
\overset{a.s.}{=} \mathbb{E}[\phi(W + \sqrt{\lambda_t^2 - \sigma_w^2} \,  {Z_1}, W)],
\label{eq:limit_num1}
 \end{split}
 \end{equation}
where $W \sim f(W)$ is independent of $Z_1$ standard Gaussian. 
Noticing that
$$\lim_m  \mathbb{E}_{{\Z_1}}[\phi_m({\w} + \sqrt{\lambda_t^2 - \sigma_w^2} \, {\Z_1}, {\w})] = \lim_m  \frac{1}{m} \sum_{i=1}^m\mathbb{E}_{{\Z_1}}[\phi({w}_i + \sqrt{\lambda_t^2 - \sigma_w^2} \, {[\Z_1]_i}, {w_i})],$$
the result follows by the SLLN (Theorem~\ref{thm:LLN}) since the terms in the sum are i.i.d., if we are able to show that $\mathbb{E}[\phi(W + \sqrt{\lambda_t^2 - \sigma_w^2} \,  {Z_1}, W)]$ is finite.
By Lemma \ref{prop_Lipschitz}  it can be shown that if $\phi: \mathbb{R}^2 \rightarrow \mathbb{R}$ is PL(2), then there is a constant $L' = \max\{2L, |\phi(\mathbf{0})|\} > 0$ such that for all $\mathbf{x}\in\mathbb{R}^2$ : $|\phi(\x)|\leq L'(1+||\x||^2/2).$ Using this,
\begin{equation}
\begin{aligned}
\label{eq:W_equation}
\lvert \phi(W + \sqrt{\lambda_t^2 - \sigma_w^2} \,  {Z_1}, W)\lvert 
&\leq L'(1+ \frac{1}{2}\lvert  \lvert (W + \sqrt{\lambda_t^2 - \sigma_w^2} \,  {Z_1}, W)   \lvert  \lvert^2)\\
&\leq L'(1+\frac{3}{2} W^2 +  (\lambda_t^2 - \sigma_w^2)  Z_1^2),
\end{aligned}
\end{equation}
where we have used the
property that for any $a_1, a_2$ scalars 
and any $r >0$, $(|a_1| +|a_2|)^r \leq 2^{r-1}(|a_1|^r +|a_2|^r)$. 
Thus,
\begin{align*} 
\lvert \lvert (W + \sqrt{\lambda_t^2 - \sigma_w^2} \,  {Z_1}, W)  \lvert \lvert^2 &= (W + \sqrt{\lambda_t^2 - \sigma_w^2} \, {Z_1})^2 + W^2 \leq 3W^2 + 2(\lambda_t^2 - \sigma_w^2)  Z_1^2.
\end{align*}
Now we have shown using \eqref{eq:W_equation},
$$\mathbb{E}\lvert \phi(W + \sqrt{\lambda_t^2 - \sigma_w^2} \,  {Z_1}, W)\lvert 
\leq  L' (1+\frac{3}{2} \mathbb{E}[W^2] + (\lambda_t^2 - \sigma_w^2)  \mathbb{E}[Z_1^2]) < \infty,$$
where the final inequality uses the boundedness of the moments of the noise in assumption \textbf{(A2)}.

The second result of \eqref{eq:main} requires a bit more care as it is not immediate that the function $\gamma_n: \mathbb{R}^{2n} \rightarrow \mathbb{R}$, defined as $\gamma_n(\mathbf{a}, \mathbf{b}) := \psi_n(\mathbf{a}, \mathbf{b}, \widetilde{\x})
\label{eq:gamma_def}$ for a sequence of SI $\{\widetilde{\x}\}_n$, is uniformly PL(2) as needed to apply Theorem~\ref{thm:Berthier}. The next step of the proof
deals with carefully handling this issue. We note that once we have shown that
\begin{equation}
\label{eq:to_show1}
\lim_n ( \psi_n(\x^t + {\A}^T {\ramp}^t,\x, \widetilde{\x}) -  \mathbb{E}_{{\Z_2}}[\psi_n({\x} + \lambda_t {\Z_2}, {\x}, \widetilde{\x})] )\overset{p}{=}0,
\end{equation}
then the last step showing that
\begin{align*}
&\lim_n  \mathbb{E}_{{\Z_2}}[\psi_n({\x} + \lambda_t {\Z_2}, {\x}, \widetilde{\x})] = \mathbb{E}[\psi({X} + \lambda_t {Z_2}, {X},\widetilde{X})],
 \end{align*}
follows by the SLLN along with Assumption \textbf{(A3)} as in \eqref{eq:limit_num1} - \eqref{eq:W_equation}.  

However, the function $\gamma_n$ is not obviously uniformly PL(2)  since an upper bound on $|\psi_n(\mathbf{a}, \widetilde{\mathbf{a}}, \widetilde{\x}) - \psi_n(\mathbf{b}, \widetilde{\mathbf{b}}, \widetilde{\x})|$ necessarily has an $||\widetilde{\x}||/\sqrt{n}$ factor, and therefore the use of Theorem~\ref{thm:Berthier} to give result \eqref{eq:to_show1} needs to be justified in more detail.  However, this is mainly a technicality as $||\widetilde{\x}||/\sqrt{n}$ is bounded by a constant (independent of $n$) with high probability. 

To show \eqref{eq:to_show1}, we would like to show that for any $\epsilon > 0$, 
\begin{equation}
P\Big(\Big \lvert  \psi_n(\x^t + {\A}^T {\ramp}^t,\x, \widetilde{\x}) -  \mathbb{E}_{{\Z_2}}[\psi_n({\x} + \lambda_t {\Z_2}, {\x}, \widetilde{\x})] \Big \lvert > \epsilon\Big) \rightarrow 0,
\label{eq:conv_in_prob}
\end{equation}
as $n \rightarrow \infty$.  Define a pair of events $\mathcal{T}_n(\epsilon)$ and $\mathcal{B}_n(C)$ as 
\[\mathcal{T}_n(\epsilon) := \Big\{\Big\lvert  \psi_n(\x^t + {\A}^T {\ramp}^t,\x, \widetilde{\x}) -  \mathbb{E}_{{\Z_2}}[\psi_n({\x} + \lambda_t {\Z_2}, {\x}, \widetilde{\x})] \Big \lvert > \epsilon \Big\},\]
and for constant $C>0$ independent of $n$,
$$\mathcal{B}_n(C) := \Big\{\widetilde{\x} \in \mathbb{R}^n : \frac{1}{\sqrt{n}}||\widetilde{\x}|| < C\Big\}.$$
Then demonstrating \eqref{eq:conv_in_prob} means showing, for any $\epsilon > 0
$, that $\lim_n P(\mathcal{T}_n(\epsilon)) = 0$.  Note that,
\begin{equation}
\begin{split}
P(\mathcal{T}_n(\epsilon)) &=  P(\mathcal{T}_n(\epsilon) \text{ and } \mathcal{B}_n(C)) + P(\mathcal{T}_n(\epsilon)  \text{ and not } \mathcal{B}_n(C) ) \leq  P(\mathcal{T}_n(\epsilon) \lvert \mathcal{B}_n(C))+ P(\text{not } \mathcal{B}_n(C)).
\label{eq:T_decompose}
\end{split}
\end{equation}

Considering \eqref{eq:T_decompose}, we argue that the first term on the right side approaches $0$ as $n$ grows due to Theorem~\ref{thm:Berthier}. First notice that, in the notation just introduced, Theorem~\ref{thm:Berthier} implies that as $n$ grows,
\begin{equation}
P\Big(\mathcal{T}_n(\epsilon) \, \Big \lvert \, \mathcal{B}_p(C) \text{ for all integers }  p \Big)  \rightarrow 0,
\label{eq:thm_application}
\end{equation}
where we note that there is a slight technicality relating to the fact  that we must now justify the conditions needed for the proof, namely, \textbf{(C1)}-\textbf{(C5)}, in the conditional probability space.  For the moment, we assume that \eqref{eq:thm_application} is true. By \eqref{eq:thm_application}, we find that the first term on the right side of \eqref{eq:T_decompose} approaches $0$ as $n \rightarrow \infty$ gets large since $\mathcal{T}_n(\epsilon)$ conditional on $\mathcal{B}_n(C)$ is independent of $\mathcal{B}_p(C)$ for $p \neq n$, and so 
$$P\Big(\mathcal{T}_n(\epsilon) \, \Big \lvert \, \mathcal{B}_n(C)\Big) = P\Big(\mathcal{T}_n(\epsilon) \, \Big \lvert \, \mathcal{B}_p(C)  \text{ for all integers }  p \Big).$$

Next, by choosing $C$ large enough, the second probability in \eqref{eq:T_decompose}, namely $P(\text{not } \mathcal{B}_n(C))$ goes to zero by the SLLN as $||\widetilde{\x}||^2/n$ concentrates to  the elementwise squared expectation of the SI, namely $\mathbb{E}[\widetilde{X}^2]$.

Now we argue that it is straightforward to justify the conditions  \textbf{(C1)}-\textbf{(C5)} in the conditional probability space, leading to the result in \eqref{eq:thm_application}.  Using that $\widetilde{\x}$ is independent of the random elements $\A$ and $\w$, this requires arguing that \textbf{(C3)} and \textbf{(C5)} are true, conditional on $\mathcal{B}_p(C)$ for all integers $p$.  We give a sketch of how to justify \textbf{(C3)} and \textbf{(C5)} in the conditional probability space.  

First consider \textbf{(C3)}. We argue that \textbf{(C3)} holds by demonstrating that
$$P\Big(\Big \lvert \frac{1}{n}||\x||^2 - \text{const}\Big \lvert \geq \epsilon \, \Big \lvert \, \mathcal{B}_p(C) \text{ for all integers }  p \Big)  \rightarrow 0.$$
As above, we notice that ${||\x||^2}/{n}$ is independent of $\mathcal{B}_p(C)$ for $p \neq n$ when we condition on $\mathcal{B}_n(C)$, and therefore
$$P\Big(\Big \lvert \frac{1}{n}||\x||^2 - \text{const}\Big \lvert \geq \epsilon \, \Big \lvert \, \mathcal{B}_p(C) \text{ for all integers }  p \Big) = P\Big(\Big \lvert \frac{1}{n}||\x||^2 - \text{const}\Big \lvert \geq \epsilon \, \Big \lvert \, \mathcal{B}_n(C) \Big).$$
Now we note that
\begin{equation}
\begin{aligned}
P\Big(\Big \lvert \frac{1}{n}||\x||^2 - \text{const}\Big \lvert \geq \epsilon \, \Big \lvert \, \mathcal{B}_n(C) \Big)  &= \frac{P\Big(\Big \lvert \frac{1}{n}||\x||^2 - \text{const}\Big \lvert \geq \epsilon; \mathcal{B}_n(C) \Big)}{P\Big(\mathcal{B}_n(C) \Big)}\\
&\leq \frac{P\Big(\Big \lvert \frac{1}{n}||\x||^2 - \text{const}\Big \lvert \geq \epsilon\Big)}{P\Big(\mathcal{B}_n(C) \Big)}.
\label{eq:string1}
\end{aligned}
\end{equation}
As we argued earlier, $P(\text{not } \mathcal{B}_n(C))  \rightarrow 0$ by the SLLN and therefore $P(\mathcal{B}_n(C)) \rightarrow 1$.  So for some $N_0$, we have, say, $P(\mathcal{B}_n(C)) > 1/2$ for all $n > N_0$.  Then the above goes to $0$ since (unconditionally) ${||\x||^2}/{n}$ concentrates on a constant by the SLLN as argued in Section \ref{sec:step2}.

Finally, the same strategy can work to prove \textbf{(C5)}.  We would like to show, for example, that
$$P\Big(\Big \lvert \frac{1}{n}\sum_{i=1}^n \mathbb{E}_{\Z}[x_i \widetilde{\eta}^t_n(x_i + Z_i)] - \text{const}\Big \lvert \geq \epsilon \, \Big \lvert \, \mathcal{B}_p(C) \text{ for all integers }  p \Big)  \rightarrow 0.$$
Noticing that $\frac{1}{n}\sum_{i=1}^n \mathbb{E}_{\Z}[x_i \widetilde{\eta}^t_n(x_i + Z_i)]  = \frac{1}{n}\sum_{i=1}^n \mathbb{E}_{\Z}[x_i \eta^t_n(x_i + Z_i, \widetilde{x}_i)] $ is independent of $\mathcal{B}_p(C)$ for $p \neq n$ when we condition on $\mathcal{B}_n(C)$, we therefore have that
\begin{align*}
&P\Big(\Big \lvert \frac{1}{n}\sum_{i=1}^n \mathbb{E}_{\Z}[x_i \widetilde{\eta}^t_n(x_i + Z_i)]  - \text{const}\Big \lvert \geq \epsilon \, \Big \lvert \, \mathcal{B}_p(C) \text{ for all integers }  p \Big) \\
&\hspace{3cm} = P\Big(\Big \lvert \frac{1}{n}\sum_{i=1}^n \mathbb{E}_{\Z}[x_i \widetilde{\eta}^t_n(x_i + Z_i)]  - \text{const}\Big \lvert \geq \epsilon \, \Big \lvert \, \mathcal{B}_n(C) \Big).
\end{align*}
As in \eqref{eq:string1}, 
\begin{align*}
& P\Big(\Big \lvert \frac{1}{n}\sum_{i=1}^n \mathbb{E}_{\Z}[x_i \widetilde{\eta}^t_n(x_i + Z_i)]  - \text{const}\Big \lvert \geq \epsilon \, \Big \lvert \, \mathcal{B}_n(C) \Big) \leq \frac{P\Big(\Big \lvert \frac{1}{n}\sum_{i=1}^n \mathbb{E}_{\Z}[x_i \widetilde{\eta}^t_n(x_i + Z_i)]  - \text{const}\Big \lvert \geq \epsilon \Big)}{P\Big(\mathcal{B}_n(C) \Big)},
\end{align*}
and the result follows by what was previously shown in Section~\ref{sec:step2}.

\section{Proof of Theorem~\ref{thm:non-separable}}\label{main_proof_non-separable}
In this section, we prove Theorem~\ref{thm:non-separable}. 
Again, we use Berthier \emph{et al}.~\cite[Theorem 14]{Berthier2017}, restated above in Theorem~\ref{thm:Berthier}. The proof follows similarly to that of Theorem~\ref{thm:SE} given in Section~\ref{main_proof}.  The main difference is carefully handling the fact that allowing a general joint distribution for $(\x, \widetilde{\x})$ means that there are now (possibly) non-trivial dependencies between the elements of $\x$ and $\widetilde{\x}$. Before we get to the proof, we state and prove a lemma that provides a weak law of large numbers when the elements of the sum are dependent\cite{Cacoullos1989}.
\begin{lemma}
\label{lem:WLLN}
Let $X_1, X_2,...$ be a sequence of RVs, each having the same mean $\mu$, with
$\mathrm{Var}(X_i)\leq c < \infty$, and $\mathrm{Cov}(X_i,X_j)\to 0$ when $|i-j|\to \infty$. Then
\begin{equation}\label{eq:WLLN}
\lim_{n\to\infty}P\left(\left|\frac{1}{n}(X_1+X_2+...+X_n)-\mu \right|\geq\epsilon\right)=0.
\end{equation}
In words, the partial averages $\frac{1}{n}(X_1+X_2+...+X_n)$ converge in probability
to $\mu < \infty$.
\end{lemma}
\begin{proof}
Denote $S_n:=X_1+X_2+...+X_n$.  
For any fixed $n$, consider all pairs $1 \leq i,j \leq n$ such that
$|i-j|>\lfloor \sqrt{n} \rfloor$ and denote a set $\textrm{pairs}(n) = \{(i,j) :1 \leq i,j \leq n \text{ and } |i-j|>\lfloor \sqrt{n} \rfloor\}$.  Then let $\varepsilon(n) := \max_{(i,j) \in \textrm{pairs}(n)} |\mathrm{Cov}(X_i,X_j)|$.  Then we know that $\lim_{n \rightarrow \infty} \epsilon(n) = 0.$
Thus, for any fixed $n$, labeling $n_0 = \lfloor \sqrt{n} \rfloor$,
\begin{equation}
\begin{split}
&\mathrm{Var}\left(S_n \right) =\sum_{i=1}^n\mathrm{Var}(x_i)+2\sum_{i=1}^n\sum_{j=i+1}^n\mathrm{Cov}(X_i,X_j) \leq nc+2\sum_{i=1}^n\sum_{j=i+1}^n\mathrm{Cov}(X_i,X_j)\\
&\leq nc+2\sum_{i=1}^n\sum_{j=i+1}^{i+n_0}\mathrm{Cov}(X_i,X_j)+2\sum_{i=1}^n\sum_{j=i+1+n_0}^{n}\mathrm{Cov}(X_i,X_j) \leq nc+2nn_0c+2n^2\varepsilon(n),
\end{split}
\end{equation}
where the final inequality follows Cauchy-Schwarz as follows: $|\mathrm{Cov}(X_i,X_j)|\leq\sqrt{\mathrm{Var}(X_i)\mathrm{Var}(X_j)}\leq c$. Therefore, 
\begin{equation}
\begin{split}
\mathrm{Var}\Big(\frac{S_n} {n}\Big)=\frac{1}{n^2}\mathrm{Var}(S_n) \leq \frac{c}{n}+\frac{2c}{\sqrt{n}}+2\varepsilon(n).
\end{split}
\end{equation}
Since $\varepsilon(n)$ goes to $0$ with $n$, we see that  $\mathrm{Var}(S_n/n)\to0$ as $n \to \infty$. 
Then by Chebyshev's inequality,
\[\lim_{n\to\infty}P\left(\Big|\frac{S_n}{n}-\mu\Big|>\epsilon\right)\leq\lim_{n\to\infty}\frac{1}{\epsilon^2}\mathrm{Var}\Big(\frac{S_n}{n}\Big)=0.\]
\end{proof}

To begin the proof,
we demonstrate that our assumptions $\textbf{(B1)} - \textbf{(B5)}$ stated in Section~\ref{main_result} are enough to satisfy the technical conditions $\textbf{(C1)} - \textbf{(C5)}$ needed to apply Theorem~\ref{thm:Berthier}.  
First, assumptions $\textbf{(B1)}$ and $\textbf{(B5)}$ are identical to $\textbf{(C1)},$ and $\textbf{(C5)}$, respectively.

Next consider assumption $\textbf{(C2)}$. The sequence of non-separable denoisers $\widetilde{\eta}_{n}^t(\x)$ are those that apply the denoiser $g_t$ defined in \eqref{eq:eta_2} as follows: $\widetilde{\eta}_{n}^t({\x}) := g_t({\x}, \widetilde{{\x}})$.  From $\textbf{(B2)}$, we know that $\{g_t(\x, \widetilde{\x})\}_{t \geq 0}$ are uniformly Lipschitz continuous, for length-$n$ vectors $\x_1, \x_2$, and fixed SI $\widetilde{\x}$,
\[\frac{||g_t({\x_1, \widetilde{\x}}) - g_t({\x_2,\widetilde{\x}})||^2}{n} \leq L^2 \frac{||(\x_1, \widetilde{\x}) - (\x_2, \widetilde{\x})||^2}{2n} = L^2 \frac{||\x_1- \x_2||^2}{2n},\]
for a constant $L > 0$ that does not depend on $n$.  Therefore,
\begin{equation*}
\begin{split}
\frac{||\widetilde{\eta}_{n}^t({\x_1}) - \widetilde{\eta}_{n}^t({\x_2})||^2}{ n}&=\frac{||g_t({\x_1, \widetilde{\x}}) - g_t({\x_2,\widetilde{\x}})||^2}{n} \\
&\leq L^2 \frac{||\x_1 - \x_2||^2}{2n} = \frac{L^2}{2} \frac{||\x_1 - \x_2||^2}{n}.
\end{split}
\end{equation*}
The Lipschitz constant 
$L^2/2$ does not depend on $n$, so $\widetilde{\eta}_{n}^t({\cdot})$ is uniformly Lipschitz.

Next we show that assumption $\textbf{(C3)}$ can be met. From assumption $\textbf{(B3)}$, each entry of the signal has the same finite second moment, which we will call $\mathbb{E}[X_1^2]$.  Moreover, each entry $i \in \{1, 2, \ldots, n\}$ has finite variance $\mathrm{Var}(X_i^2)$ since $\mathrm{Var}(X_i^2) = \mathbb{E}(X_i^4) - [\mathbb{E}(X_i^2)]^2 < \infty$ and  $ [\mathbb{E}(X_i^2)]^2 \leq \mathbb{E}(X_i^4) < \infty$ where the first inequality follows by Jensen's Inequality and the second by assumption. 
Then applying Lemma \ref{lem:WLLN},
\begin{equation*}
\lim_{n\to\infty}\frac{1}{n}||\textbf{x}||_2^2=\lim_{n\to\infty}\frac{1}{n}\sum_{i=1}^n x_i^2
\overset{p}{=}\mathbb{E}[X_1^2],
\end{equation*}
where we have used that $\mathrm{Cov}(X_i^2, X_j^2) \rightarrow 0$ as $|i-j| \rightarrow \infty$.

Next, consider the assumption $\textbf{(C4)}$. From $\textbf{(B4)}$, the noise $\w$ sampled i.i.d.\ from $\sim f(W)$ with finite $\mathbb{E}[W^2]$. Therefore, by SLLN, 
\begin{equation*}
\lim_{m\to\infty}\frac{1}{m}||\w||_2^2=\lim_{m\to\infty}\frac{1}{m}\sum_{i=1}^m w_i^2
=\sigma_w^2 < \infty.
\end{equation*}

Therefore,  the technical conditions $\textbf{(C1)} - \textbf{(C5)}$ are satisfied and we may apply Theorem~\ref{thm:Berthier}.
Now, the first result in \eqref{eq:main2}, namely the asymptotic result for $\kappa_m$, a uniformly pseudo-Lipschitz function, follows immediately by applying Theorem~\ref{thm:Berthier} using $\rho_m = \kappa_m$. 

We consider the second result in \eqref{eq:main2}, namely the asymptotic result for $\nu_n$ uniformly pseudo-Lipschitz: we want to show
\[\lim_n \nu_n({\x}^{t} + {\A}^T {\ramp}^t, {\x},\widetilde{\x}) \overset{p}{=} \lim_n \mathbb{E}_{{\Z_2}}\left[\nu_n({\x} + \lambda_t {\Z_2}, {\x},\widetilde{\x})\right].\]
We note that applying Theorem~\ref{thm:Berthier} using $\gamma_n: \mathbb{R}^{2n} \rightarrow \mathbb{R}$ defined as $\gamma_n(\mathbf{a}, \mathbf{b}) := \nu_n(\mathbf{a}, \mathbf{b}, \widetilde{\x})
\label{eq:gamma_def2}$ for a sequence of SI $\{\widetilde{\x}\}_n$ would give the desired result, however it is not immediately obvious that such a function $\gamma_n$ is uniformly PL(k) as needed to apply Theorem~\ref{thm:Berthier} even under the assumption that $\nu_n$ is PL(k).  Justifying that we can still apply Theorem~\ref{thm:Berthier} to the function $\gamma_n$ follows similarly to the same step in the proof in Section \ref{Step 3}.  Namely, the problem lies in the fact that an upperbound on $||\gamma_n(\mathbf{a}, \mathbf{b}) - \gamma_n(\widetilde{\mathbf{a}}, \widetilde{\mathbf{b}})||/\sqrt{2n}$ necessarily has an  $(||\widetilde{\x}||/\sqrt{n})^{k-1}$ factor meaning that it is not obviously uniformly pseudo-Lipschitz.  We deal with this by showing that  $(||\widetilde{\x}||/\sqrt{n})^{k-1}$ is bounded by a constant, independent of $n$, with high probability.

Define events $\mathcal{T}'_n(\epsilon)$ and $\mathcal{B}'_n(C')$ as follows,
\[\mathcal{T}'_n(\epsilon) := \{\lvert \nu_n(\x^t + {\A}^T {\ramp}^t,\x, \widetilde{\x}) -  \mathbb{E}_{{\Z_2}}[\nu_n({\x} + \lambda_t {\Z_2}, {\x}, \widetilde{\x})]  \lvert > \epsilon \},\] and for constant $C'>0$ independent of $n$, 
\[\mathcal{B}'_n(C') := \{\widetilde{\x} \in \mathbb{R}^n : \big(||\widetilde{\x}||/\sqrt{n}\big)^{k-1} < C'\}.\] 
Then to prove  the second result in \eqref{eq:main2}, we would like to prove that $\lim_{n\to \infty} P(\mathcal{T}'_n(\epsilon)) = 0$ for any $\epsilon > 0$.
Again, as in \eqref{eq:T_decompose}, we have an upper bound on $P(\mathcal{T}'_n(\epsilon))$ given by
\begin{equation}\label{eq:T_decompose_noniid}
\begin{split}
P(\mathcal{T}'_n(\epsilon))\leq  P(\mathcal{T}'_n(\epsilon) \, \lvert \, \mathcal{B}'_n(C'))+ P(\text{not } \mathcal{B}'_n(C')).
\end{split}
\end{equation}
Now considering \eqref{eq:T_decompose_noniid}, the first term on the right side approaches $0$ as $n$ grows by Theorem~\ref{thm:Berthier}. To see this, notice that, in the notation just introduced, Theorem~\ref{thm:Berthier} implies that as $n$ grows,
\begin{equation}
P\Big(\mathcal{T}'_n(\epsilon) \, \Big \lvert \, \mathcal{B}'_p(C') \text{ for all integers }  p \Big)  \rightarrow 0.
\label{eq:thm_application2}
\end{equation}
As in the proof in Section \ref{Step 3}, in order for the above to be true, we must additionally justify the conditions needed for the proof, namely, \textbf{(C1)}-\textbf{(C5)}, in the conditional probability space.  We assume that \eqref{eq:thm_application2} is true for now.

Now we use this to show that the probability in \eqref{eq:T_decompose_noniid} goes to $0$.  The first term, $P(\mathcal{T}'_n(\epsilon) \, \lvert \, \mathcal{B}'_n(C'))$, approaches $0$ as $n$ gets large due to the fact that $\mathcal{T}'_n(\epsilon)$ is independent of $\mathcal{B}'_p(C')$ for $p\neq n$, when conditioning on $\mathcal{B}'_p(C')$, or in other words,
\[P\Big(\mathcal{T}'_n(\epsilon) \, \Big \lvert \, \mathcal{B}'_p(C') \text{ for all integers }  p \Big)  = P\Big(\mathcal{T}'_n(\epsilon) \, \Big \lvert \, \mathcal{B}'_n(C')\Big). \] 
Next, by choosing $C'$ large enough, the second probability, $P(\text{not } \mathcal{B'}_n(C'))$, goes to zero by Lemma \ref{lem:WLLN}, because $( ||\widetilde{\x}||^2/n)^{(k-1)/2}$ concentrates to $(\mathbb{E}[\widetilde{X}_1^2])^{{(k-1)}/{2}}$ using the assumptions in \textbf{(B3)} and arguments like that justifying \textbf{(C3)} above.

Now we deal with the fact that we need to justify the assumptions for the proof\textbf{(C1)}-\textbf{(C5)}, in the conditional probability space.  This is done in a manner largely similar to what was done in the proof in Section \ref{Step 3}. First, since $\widetilde{\x}$ is independent of $\A$ and $\w$, so is the conditioning event $\mathcal{B'}_p(C')$ for all integers $p$.  Therefore, the arguments justifying \textbf{(C1)}, \textbf{(C2)}, and \textbf{(C4)} remain the same so in what follows we only consider \textbf{(C3)} and \textbf{(C5)}. Now we give a sketch of how to justify \textbf{(C3)} and \textbf{(C5)} in the conditional probability space.

First consider \textbf{(C3)}. 
Notice that ${||\x||^2}/{n}$ is independent of $\mathcal{B'}_p(C')$ for $p \neq n$ when we condition on $\mathcal{B'}_n(C')$, and therefore
$$P\Big(\Big \lvert \frac{1}{n}||\x||^2 - \text{const}\Big \lvert \geq \epsilon \, \Big \lvert \, \mathcal{B'}_p(C') \text{ for all integers }  p \Big) = P\Big(\Big \lvert \frac{1}{n}||\x||^2 - \text{const}\Big \lvert \geq \epsilon \, \Big \lvert \, \mathcal{B'}_n(C') \Big).$$
Then, as in \eqref{eq:string1},
\begin{align}
&P\Big(\Big \lvert \frac{1}{n}||\x||^2 - \text{const}\Big \lvert \geq \epsilon \, \Big \lvert \, \mathcal{B'}_n(C') \Big)  \leq \frac{P\Big(\Big \lvert \frac{1}{n}||\x||^2 - \text{const}\Big \lvert \geq \epsilon\Big)}{P\Big(\mathcal{B'}_n(C') \Big)}.
\label{eq:string2}
\end{align}
As we argued earlier, $P(\text{not } \mathcal{B'}_n(C'))  \rightarrow 0$ by Lemma \ref{lem:WLLN} and therefore $P(\mathcal{B'}_n(C')) \rightarrow 1$.  Then the above \eqref{eq:string2} converges to zero, because ${||\x||^2}/{n}$ converges in probability (unconditionally) to a constant, as shown by Lemma \ref{lem:WLLN}.

Finally, we sketch how to show the first result in \textbf{(C5)}.  In particular, we need to show
$$P\Big( \Big \lvert \frac{1}{n} \mathbb{E}_{\Z}[\x^T \widetilde{\eta}^t_n(\x+\Z)] - \text{ const} \Big \lvert > \epsilon \, \Big \lvert \, \mathcal{B'}_p(C') \text{ for all integers }  p \Big)  \rightarrow 0.$$
Noticing that $\frac{1}{n} \mathbb{E}_{\Z}[\x^T \widetilde{\eta}^t_n(\x+\Z)] = \frac{1}{n} \mathbb{E}_{\Z}[\x^T {g}_t(\x+\Z, \widetilde{\x})]$ is independent of $\mathcal{B'}_p(C')$ for $p \neq n$ when we condition on $\mathcal{B'}_n(C')$, we have
\begin{align*}
&P\Big( \Big \lvert \frac{1}{n} \mathbb{E}_{\Z}[\x^T \widetilde{\eta}^t_n(\x+\Z)] - \text{ const} \Big \lvert > \epsilon \, \Big \lvert \, \mathcal{B'}_p(C') \text{ for all integers }  p \Big)  \\
&\hspace{3cm} = P\Big( \Big \lvert \frac{1}{n} \mathbb{E}_{\Z}[\x^T \widetilde{\eta}^t_n(\x+\Z)] - \text{ const} \Big \lvert > \epsilon \, \Big \lvert \, \mathcal{B'}_n(C') \Big) .
\end{align*}
As in \eqref{eq:string2}, we can upper bound this as follows:
\begin{align*}
& P\Big( \Big \lvert \frac{1}{n} \mathbb{E}_{\Z}[\x^T \widetilde{\eta}^t_n(\x+\Z)] - \text{ const} \Big \lvert > \epsilon \, \Big \lvert \, \mathcal{B'}_n(C') \Big)  \leq \frac{P\Big(\Big \lvert \frac{1}{n} \mathbb{E}_{\Z}[\x^T \widetilde{\eta}^t_n(\x+\Z)] - \text{ const} \Big \lvert > \epsilon \Big)}{P\Big(\mathcal{B'}_n(C') \Big)},
\end{align*}
and from $\textbf{(B5)}$, the limit $\lim_{n\to\infty} \frac{1}{n} \mathbb{E}_{\Z}[\x^T \widetilde{\eta}^t_n(\x+\Z)]$ exists and is finite, so that $P\Big(\Big \lvert \frac{1}{n} \mathbb{E}_{\Z}[\x^T \widetilde{\eta}^t_n(\x+\Z)] - \text{ const} \Big \lvert > \epsilon \Big) \rightarrow 0$.

\section{Conclusion}
In this paper, we prove that under various technical conditions, approximate message passing using side information (AMP-SI), as quantified by uniformly pseudo-Lipschitz loss, can be characterized by a state evolution (SE) when the signal and SI have either i.i.d.\ or 
non-i.i.d.\ entries. Our technical conditions require the measurement matrix to be i.i.d.\ Gaussian, 
and the joint distribution of the input signal and SI should satisfy some finite moment constraints.

The main difficulty is to establish that Bayes AMP-SI denoisers are uniformly Lipschitz when also conditioning on the SI. When the signal and SI have i.i.d.\ entries, the Bayes AMP-SI denoiser is 
separable and, in certain cases, Lipschitz continuous. As illustrative examples, we provided conditional denoisers 
for the AMP-SI algorithm and its corresponding SE in two signal and SI models: the Gaussian signal/Gaussian noise model and Bernoulli-Gaussian signal/Gaussian noise model. 
When there exist dependencies between the signal and SI pair, the Bayes AMP-SI denoiser is non-separable, 
but in many cases remains uniformly Lipschitz, as we demonstrate for the block-sparse signal model.

\section*{Acknowledgment}
We thank You (Joe) Zhou for insightful conversations and valuable advice. 

\bibliographystyle{IEEEtran}
\bibliography{bibliography}

\appendix

\section{Various Technical Results} \label{app:technical}

\subsection{Proof of Lemma~\ref{prop_Lipschitz}}
\begin{proof}
Letting $\mathbf{0}\in\mathbb{R}^n$ denote the all-zero vector, then $\mathbf{a}\in\mathbb{R}^m$ defined as $\mathbf{a} := \phi(\mathbf{0})$ is a constant vector. Using the Cauchy-Schwarz inequality, it can be seen that,
\begin{equation*}
\begin{aligned}
\frac{1}{m} \left(||\phi(\x)||-||\mathbf{a}||\right)^2 \leq \frac{1}{m} \left|\left|\phi(\x)-\phi(\mathbf{0})\right|\right|^2.
\end{aligned}
\end{equation*}
Therefore, by Definition~\ref{def:PLfunc} there exists some constant $L > 0$ such that
\begin{equation*}
\begin{aligned}
\frac{1}{\sqrt{m}}\Big(||\phi(\x)||-||\mathbf{a}||\Big)\leq \frac{1}{\sqrt{m}}\Big|||\phi(\x)||-||\mathbf{a}||\Big|\leq\frac{1}{\sqrt{m}}\left|\left|\phi(\x)-\phi(\mathbf{0})\right|\right|\leq L\Big(1+ \Big(\frac{||\x||}{\sqrt{n}}\Big)^{k-1}\Big) \frac{||\x||}{\sqrt{n}}.
\end{aligned}
\end{equation*}
Rearranging, we find
\begin{equation*}
\begin{aligned}
 \frac{||\phi(\x)||}{\sqrt{m}}&\leq L\Big(1+ \Big(\frac{||\x||}{\sqrt{n}}\Big)^{k-1}\Big) \frac{||\x||}{\sqrt{n}}+\frac{||\mathbf{a}||}{\sqrt{m}}\leq L\Big(\frac{||\x||}{\sqrt{n}}\Big)+ L\Big(\frac{||\x||}{\sqrt{n}}\Big)^{k} +\frac{||\mathbf{a}||}{\sqrt{m}}.
\end{aligned}
\end{equation*}
From the bound just above, letting $L' = \max\{2L, ||\mathbf{a}||/\sqrt{m}\}$, we consider two cases.  First, if $||\x||/\sqrt{n} \leq 1$, then
\begin{equation*}
\begin{aligned}
 \frac{||\phi(\x)||}{\sqrt{m}}&\leq L\Big(\frac{||\x||}{\sqrt{n}}\Big)+ L\Big(\frac{||\x||}{\sqrt{n}}\Big)^{k} +\frac{||\mathbf{a}||}{\sqrt{m}} \leq L\Big(1 +\Big(\frac{||\x||}{\sqrt{n}}\Big)^{k}\Big) +\frac{||\mathbf{a}||}{\sqrt{m}} \leq L' \Big( 1 + \Big(\frac{||\x||}{\sqrt{n}}\Big)^{k} \Big),
\end{aligned}
\end{equation*}
giving the desired result. Next, if $||\x||/\sqrt{n} > 1$ we find the same bound:
\begin{equation*}
\begin{aligned}
 \frac{||\phi(\x)||}{\sqrt{m}}&\leq L\Big(\frac{||\x||}{\sqrt{n}}\Big)+ L\Big(\frac{||\x||}{\sqrt{n}}\Big)^{k} +\frac{||\mathbf{a}||}{\sqrt{m}} \leq 2L\Big(\frac{||\x||}{\sqrt{n}}\Big)^{k} +\frac{||\mathbf{a}||}{\sqrt{m}}  \leq L' \Big( 1 + \Big(\frac{||\x||}{\sqrt{n}}\Big)^{k} \Big).
\end{aligned}
\end{equation*}
\end{proof}

\subsection{Technical Details for the Proof of Corollary~\ref{cor:SE-nonsep}} \label{Appendix_Cor2.2.1}

Here we aim to show that $\nu_n^2(\x, \y, \z) = \frac{1}{n}||g_t(\x,\z)-\y||^2$ is a uniformly PL(2) function.  Recall that this means we want to demonstrate the following upper bound,
\begin{equation*}
\begin{split}
&|\nu_n^2(\x, \y, \z)-\nu_n^2(\widetilde{\x}, \widetilde{\y}, \widetilde{\z})| \leq L\left(1+\frac{||(\x,\y,\z)||}{\sqrt{3n}}+\frac{||(\widetilde{\x},\widetilde{\y},\widetilde{\z})||}{\sqrt{3n}}\right)\frac{||(\x,\y,\z)-(\widetilde{\x},\widetilde{\y},\widetilde{\z})||}{\sqrt{3n}},
\end{split}
\end{equation*}
where $L>0$ is a positive constant that does not depend on $n$.

We first note,
\begin{equation}
\begin{split}
&|\nu_n^2(\x, \y, \z)-\nu_n^2(\widetilde{\x}, \widetilde{\y}, \widetilde{\z})| = \frac{1}{n}\left|||g_t(\x,\z)-\y||^2 -||g_t(\widetilde{\x},\widetilde{\z})- \widetilde{\y}||^2\right|\\
&=\frac{1}{n}\left|\sum_{i=1}^{n}\left(\left([g_t(\x,\z)]_i-y_i\right)^2-\left([g_t(\widetilde{\x},\widetilde{\z})]_i-\widetilde{y}_i\right)^2\right)\right|\leq\frac{1}{n}\sum_{i=1}^{n}\left|\left([g_t(\x,\z)]_i-y_i\right)^2-\left([g_t(\widetilde{\x},\widetilde{\z})]_i-\widetilde{y}_i\right)^2\right|\\
&=\frac{1}{n}\sum_{i=1}^{n}\Big|\left([g_t(\x,\z)]_i-y_i\right)+\left([g_t(\widetilde{\x},\widetilde{\z})]_i-\widetilde{y}_i\right)\Big|\,\,\Big|\left([g_t(\x,\z)]_i-y_i\right)-\left([g_t(\widetilde{\x},\widetilde{\z})]_i-\widetilde{y}_i\right)\Big|.
\label{eq:cor_bound1}
\end{split}
\end{equation}
By Cauchy-Schwarz, $(|a_1| +|a_2|)^r \leq 2^{r-1}(|a_1|^r +|a_2|^r)$ for any $r >0$, and thus
\begin{equation}
\begin{split}
&\Big(\sum_{i=1}^{n}\Big|\left([g_t(\x,\z)]_i-y_i\right)+\left([g_t(\widetilde{\x},\widetilde{\z}\right)]_i-\widetilde{y}_i)\Big|\,\,\Big|\left([g_t(\x,\z)]_i-[g_t(\widetilde{\x},\widetilde{\z})]_i\right)-(y_i-\widetilde{y}_i)\Big|\Big)^2\\
&\leq \Big[\sum_{i=1}^{n}\left(([g_t(\x,\z)]_i-y_i)+([g_t(\widetilde{\x},\widetilde{\z})]_i-\widetilde{y}_i)\right)^2 \Big] \Big[ \sum_{j=1}^{n}\left(([g_t(\x,\z)]_j-[g_t(\widetilde{\x},\widetilde{\z})]_j)-(y_j-\widetilde{y}_j)\right)^2 \Big]\\
&\leq \Big[4\sum_{i=1}^{n}\left(([g_t(\x,\z)]_i^2 +y_i^2)+ ([g_t(\widetilde{\x},\widetilde{\z})]_i^2 + \widetilde{y}_i^2)\right) \Big] \Big[ 2\sum_{j=1}^{n}\left(([g_t(\x,\z)]_j-[g_t(\widetilde{\x},\widetilde{\z})]_j)^2+(y_j-\widetilde{y}_j)^2\right)\Big] \\
&= 8 \left(||g_t(\x,\z)||^2+||g_t(\widetilde{\x},\widetilde{\z})||^2+||\y||^2+||\widetilde{\y}||^2\right)\left(||g_t(\x,\z)-g_t(\widetilde{\x},\widetilde{\z})||^2 +||\y-\widetilde{\y}||^2\right).
\label{eq:cor_bound2}
\end{split}
\end{equation}
Plugging \eqref{eq:cor_bound2} into \eqref{eq:cor_bound1}, we have shown
\begin{equation}
\begin{split}
&|\nu_n^2(\x, \y, \z)-\nu_n^2(\widetilde{\x}, \widetilde{\y}, \widetilde{\z})|^2 \\
&\leq 8 \left(\frac{||g_t(\x,\z)||^2}{n}+\frac{||g_t(\widetilde{\x},\widetilde{\z})||^2}{n}+\frac{||\y||^2}{n}+\frac{||\widetilde{\y}||^2}{n}\right)\left(\frac{||g_t(\x,\z)-g_t(\widetilde{\x},\widetilde{\z})||^2}{n} +\frac{||\y-\widetilde{\y}||^2}{n}\right).
\label{eq:lipschitz_res1}
\end{split}
\end{equation}
Now recall by assumption \textbf{(B2)} that the sequence of denoisers $g_t(\cdot, \cdot)$ is uniformly Lipschitz in $n$, meaning (per Definition~\ref{def:PLfunc}) that there exists a positive constant $L'>0$ that does not depend on $n$  such that,
\begin{equation}
\frac{||g_t(\x,\z)-g_t(\widetilde{\x},\widetilde{\z})||}{\sqrt{n}} \leq L' \frac{||(\x,\z)-(\widetilde{\x},\widetilde{\z})||}{\sqrt{2n}},
\label{eq:lipschitz_known1}
\end{equation}
 and, moreover, by Lemma~\ref{prop_Lipschitz} with $L'' = \max\{2L', ||g_t(\mathbf{0}, \mathbf{0})||/\sqrt{n}\}$, along with Cauchy-Schwarz,
\begin{equation}
\frac{||g_t(\x,\z)||}{\sqrt{n}} \leq L''\left(1 + \frac{||(\x,\z)||}{\sqrt{2n}}\right) \quad \implies \quad \frac{||g_t(\x,\z)||^2}{n} \leq 2L''^2\left(1 + \frac{||(\x,\z)||^2}{2n}\right).
\label{eq:lipschitz_known2}
\end{equation}
First, using \eqref{eq:lipschitz_known1} to bound the second term on the right side of \eqref{eq:lipschitz_res1} gives,
\begin{equation}
\begin{split}
&|\nu_n^2(\x, \y, \z)-\nu_n^2(\widetilde{\x}, \widetilde{\y}, \widetilde{\z})|^2\\
&\leq 8 \max\{1,L'^2\} \left(\frac{||g_t(\x,\z)||^2}{n}+\frac{||g_t(\widetilde{\x},\widetilde{\z})||^2}{n}+\frac{||\y||^2}{n}+\frac{||\widetilde{\y}||^2}{n}\right)\left(\frac{||(\x,\z)-(\widetilde{\x},\widetilde{\z})||^2}{2n} +\frac{||\y-\widetilde{\y}||^2}{n}\right) \\
&\leq 24 \max\{1,L'^2\} \left(\frac{||g_t(\x,\z)||^2}{n}+\frac{||g_t(\widetilde{\x},\widetilde{\z})||^2}{n}+\frac{||\y||^2}{n}+\frac{||\widetilde{\y}||^2}{n}\right)\frac{||(\x, \y, \z)-(\widetilde{\x}, \widetilde{\y}, \widetilde{\z})||^2}{3n}.
\label{eq:lipschitz_res2}
\end{split}
\end{equation}
Next, we use \eqref{eq:lipschitz_known2} to bound the first term on the right side of \eqref{eq:lipschitz_res2} to give the desired result:
\begin{equation}
\begin{split}
&|\nu_n^2(\x, \y, \z)-\nu_n^2(\widetilde{\x}, \widetilde{\y}, \widetilde{\z})|^2 \\
&\leq  24 \max\{1,L'^2\} \left(2L''^2 \left(1 + \frac{||(\x,\z)||^2}{2n}\right)+2L''^2\left(1 + \frac{||(\widetilde{\x},\widetilde{\z})||^2}{2n}\right)+\frac{||\y||^2}{n}+\frac{||\widetilde{\y}||^2}{n}\right)\frac{||(\x, \y, \z)-(\widetilde{\x}, \widetilde{\y}, \widetilde{\z})||^2}{3n} \\
&\leq  72 \max\{1,L'^2\}\max\{1,2L''^2\} \left(1+   \frac{||(\x,\y, \z)||^2}{3n} + \frac{||(\widetilde{\x}, \widetilde{\y},\widetilde{\z})||^2}{3n}\right) \frac{||(\x, \y, \z)-(\widetilde{\x}, \widetilde{\y}, \widetilde{\z})||^2}{3n}.
\label{eq:lipschitz_res3}
\end{split}
\end{equation}
Finally, \eqref{eq:lipschitz_res3} implies
\begin{equation*}
\begin{split}
&|\nu_n^2(\x, \y, \z)-\nu_n^2(\widetilde{\x}, \widetilde{\y}, \widetilde{\z})| \\
&\leq  6 \sqrt{2} \max\{1,L'\}\max\{1,\sqrt{2}L''\} \left(1+   \frac{||(\x,\y, \z)||^2}{3n} + \frac{||(\widetilde{\x}, \widetilde{\y},\widetilde{\z})||^2}{3n}\right)^{1/2} \frac{||(\x, \y, \z)-(\widetilde{\x}, \widetilde{\y}, \widetilde{\z})||}{\sqrt{3n}} \\
&\leq  6 \sqrt{2} \max\{1,L'\}\max\{1,\sqrt{2}L''\} \left(1+   \frac{||(\x,\y, \z)||}{\sqrt{3n}} + \frac{||(\widetilde{\x}, \widetilde{\y},\widetilde{\z})||}{\sqrt{3n}}\right) \frac{||(\x, \y, \z)-(\widetilde{\x}, \widetilde{\y}, \widetilde{\z})||}{\sqrt{3n}},
\end{split}
\end{equation*}
giving the desired result since $6 \sqrt{2} \max\{1,L'\}\max\{1,\sqrt{2}L''\} > 0$ is a constant not depending on $n$.  We note that this uses the fact that $L'' = \max\{2L', ||g_t(\mathbf{0}, \mathbf{0})||/\sqrt{n}\}$, and by assumption $||g_t(\mathbf{0}, \mathbf{0})||/\sqrt{n} \leq C$ where $C > 0$ is some constant not depending on $n$.

\section{Technical Proofs used in Examples} \label{app:examples}

\subsection{Technical Details of the Bernoulli-Gaussian Example}

We study the partial derivatives of the denoiser given in \eqref{eq:BGdenoiser} in order to use  Lemma \ref{lem:Lipschitz} to show that the denoiser is Lipschitz.

Denote
\begin{equation}
f_{a,b} :=\frac{ \sigma^2 a + \lambda_{t}^2 b}{ \sigma^2 + \lambda_{t}^2 + \sigma^2 \lambda_{t}^2}.
\label{eq:fab_def}
\end{equation}
Combining \eqref{eq:BGprobability}, \eqref{eq:Tab_def} and \eqref{eq:fab_def}, we find
$
\eta_{t}(a,b) = (1 + T_{a,b})^{-1} f_{a,b}.
$
Then,
\begin{equation}
\begin{split}
\label{eq:partial_BG_1_V1}
\Big \lvert \frac{\partial \eta_{t}(a,b) }{\partial a}  \Big \lvert &= \Big \lvert\frac{1}{1 + T_{a,b}} \Big[ \frac{\partial f_{a,b}}{\partial a} \Big] -\frac{f_{a,b}}{(1+T_{a,b})^{2}} \Big[\frac{\partial T_{a,b}}{\partial a}\Big]  \Big \lvert \\
&\leq \frac{ (1 + 2T_{a,b} )}{(1 + T_{a,b})^{2}}\Big \lvert  \frac{\partial f_{a,b}}{\partial a} \Big\lvert +  \frac{1}{(1 + T_{a,b})^{2}}  \Big \lvert   \frac{\partial (  T_{a,b} f_{a,b})}{\partial a} \Big \lvert.
\end{split}
\end{equation}
Now we show upper bounds for the two terms of \eqref{eq:partial_BG_1_V1} separately.  For the first term, 
$\frac{\partial f_{a,b}}{\partial a} \leq 1$, and
\[\frac{ (1 + 2T_{a,b} )}{(1 + T_{a,b})^{2}}\Big \lvert  \frac{\partial f_{a,b}}{\partial a} \Big\lvert \leq 1.\]

Consider the second term of \eqref{eq:partial_BG_1_V1}. First we note that
\[ \frac{1}{(1 + T_{a,b})^{2}} \Big \lvert  T_{a,b} \Big[ \frac{\partial}{\partial a} f_{a,b}\Big] + f_{a,b} \Big[\frac{\partial}{\partial a} T_{a,b}\Big]  \Big \lvert \leq  \Big \lvert \frac{\partial}{\partial a}  \Big[ T_{a,b} f_{a,b}\Big]  \Big \lvert.\]
Then from \eqref{eq:Tab_def} and \eqref{eq:fab_def},
\begin{equation*}
\begin{split}
T_{a,b} f_{a,b} 
&= \Big(\frac{1 - \epsilon}{\epsilon}\Big) \sqrt{2 \pi}(\sigma^2 a + \lambda_{t}^2 b) \rho_{ \nu_{t} }(\sigma^2 a + \lambda_{t}^2 b),
\end{split}
\end{equation*}
then using that $\frac{\partial}{\partial x} \rho_{\tau^2}(x) = - \frac{x}{\tau^2}  \rho_{\tau^2}(x)$, we have
\begin{equation}
\begin{split}
\label{eq:partial_BG_1}
\Big \lvert \frac{\partial}{\partial a}  \Big[ T_{a,b} f_{a,b}\Big]   \Big \lvert  &= \Big(\frac{1 - \epsilon}{\epsilon}\Big) \sqrt{2 \pi}\Big \lvert \sigma^2  \rho_{ \nu_{t} }(\sigma^2 a + \lambda_{t}^2 b)-\frac{\sigma^2(\sigma^2 a + \lambda_{t}^2 b)^2}{\nu_{t}}  \rho_{ \nu_{t} }(\sigma^2 a + \lambda_{t}^2 b)\Big \lvert \\
&= \Big(\frac{1 - \epsilon}{\epsilon}\Big) \frac{\sqrt{2 \pi} \sigma^2}{\nu_{t}}  \rho_{ \nu_{t} }(\sigma^2 a + \lambda_{t}^2 b) 
\Big \lvert \nu_{t} - (\sigma^2 a + \lambda_{t}^2 b)^2 \Big \lvert.
\end{split}
\end{equation}
To upper bound the above, we use $\exp\{-x\} \leq \frac{1}{1+x}$ when $x \geq 0$, and so
\[\rho_{\tau^2}(x) = \frac{1}{\sqrt{2\pi \tau^2 }} \exp\Big\{\frac{-x^2}{2 \tau^2}\Big\} \leq \sqrt{\frac{2}{\pi}} \Big(\frac{\tau}{2 \tau^2 + x^2}\Big).\]
Using this in \eqref{eq:partial_BG_1}, we find
\begin{equation}\label{eq:lastpartial_a_bound}
\begin{split}
\Big \lvert \frac{\partial}{\partial a}  \Big[ T_{a,b} f_{a,b}\Big]  \Big \lvert  &\leq \frac{2 \sigma^2}{\sqrt{ \nu_{t}}} \Big(\frac{1 - \epsilon}{\epsilon}\Big)\frac{  \lvert \nu_{t} - (\sigma^2 a + \lambda_{t}^2 b)^2 \lvert}{2  \nu_{t} + (\sigma^2 a + \lambda_{t}^2 b) ^2} \leq \frac{2 \sigma^2}{\sqrt{ \nu_{t}}} \Big(\frac{1 - \epsilon}{\epsilon}\Big) \leq \frac{2(1 - \epsilon)}{ \sigma_w  \epsilon},
\end{split}
\end{equation}
where in the final inequality we use $ \lambda_{t} \geq \sigma_w$ by \eqref{eq:BG_SE}, and 
\begin{equation}
\begin{aligned}
\frac{ \sigma^2}{\sqrt{ \nu_{t}}}& = \frac{ \sigma}{  \lambda_{t} \sqrt{ \sigma^2+\lambda_{t}^2 + \sigma^2 \lambda_{t}^2}} = \frac{ 1}{\lambda_{t} \sqrt{ 1+ \frac{\lambda_{t}^2}{\sigma^2} + \lambda_{t}^2}} \leq \frac{ 1}{  \lambda_{t}}.
\end{aligned}
\end{equation}
Using the above in \eqref{eq:partial_BG_1_V1}, we have 
\begin{equation*} 
\Big\lvert \frac{\partial}{\partial a} \eta_{t}(a,b) \Big\lvert \leq  1+  \frac{2(1 - \epsilon)}{\sigma_w \epsilon} .
\end{equation*}
As in \eqref{eq:partial_BG_1_V1}, we can show
\begin{equation*}
\begin{split}
\Big \lvert \frac{\partial}{\partial b} \eta_{t}(a,b) \Big \lvert  &\leq \frac{ (1 + 2T_{a,b} )}{(1 + T_{a,b})^{2}}\Big \lvert \Big[ \frac{\partial}{\partial b} f_{a,b}\Big] \Big\lvert 
+ \frac{1}{(1 + T_{a,b})^{2}} \Big \lvert  T_{a,b} \Big[ \frac{\partial}{\partial b} f_{a,b}\Big] + \Big[\frac{\partial}{\partial b} T_{a,b}\Big]f_{a,b}  \Big \lvert.
\end{split}
\end{equation*}
Then,
\begin{equation*}
\begin{split}
 \frac{ (1 + 2T_{a,b} )}{(1 + T_{a,b})^{2}}\Big \lvert \frac{\partial}{\partial b} f_{a,b} \Big\lvert  &\leq  1,
\end{split}
\end{equation*}
and a bound as in \eqref{eq:partial_BG_1} - \eqref{eq:lastpartial_a_bound} gives
\begin{equation*}
\begin{split}
 \frac{1}{(1 + T_{a,b})^{2}} \Big \lvert  T_{a,b} \Big[ \frac{\partial}{\partial b} f_{a,b}\Big] +f_{a,b} \Big[\frac{\partial}{\partial b} T_{a,b}\Big]  \Big \lvert \leq  \Big \lvert  \frac{\partial}{\partial b}  \Big[ T_{a,b} f_{a,b}\Big]   \Big \lvert &\leq  \frac{2  \lambda_{t}^2}{\sqrt{\nu_{t} }} \Big(\frac{1 - \epsilon}{\epsilon}\Big) \frac{ \lvert \nu_{t} - (\sigma^2 a + \lambda_{t}^2 b)^2 \lvert}{2 \nu_{t}  + (\sigma^2 a + \lambda_{t}^2 b)^2} \\
 &\leq \frac{2  \lambda_{t}^2}{\sqrt{\nu_{t} }} \Big(\frac{1 - \epsilon}{\epsilon}\Big) \leq \frac{2(1 - \epsilon)}{ \sigma  \epsilon} .
\end{split}
\end{equation*}

\subsection{Technical Details of the Block-sparse Signal Model Example}

In this section, we include the proofs of various lemmas that were used in presenting the example related to block-sparse signals.

\begin{proof}[Proof of Lemma \ref{lem:denoiser_simplify}]
We would like to study the expectation $\mathbb{E}[X_i\Big|\X_{(\ell)}+\lambda_t\Z_1=\mathbf{a},\X_{(\ell)}+\sigma \Z_2=\mathbf{b}],$ where $\X_{(\ell)}$ is a $K$-length vector with a single non-zero value equal to one and $i \in \{1, 2, \ldots, K\}$ is an index.  Moreover, $\Z_1, \Z_2$ are independent, $K$-length random vectors with standard Gaussian entries.  Since we only consider a single section $\ell \in \{1, 2, \ldots, L\}$ throughout this appendix, we drop the explicit $\ell$ subscript on $\X_{(\ell)}$, writing simply $\X$. Then we have,
\begin{equation}\label{eq:Blocksparse_denoiser}
\begin{aligned}
&\mathbb{E}\left[X_i\Big|\X +\lambda_t\Z_1=\mathbf{a},\X+\sigma \Z_2=\mathbf{b}\right]=P(X_i=1, X_j=0; j \in \ell, \, j \neq i \, | \, \X+\lambda_t\Z_1=\mathbf{a}, \X+\sigma \Z_2=\mathbf{b})\\
&=\frac{P(X_i=1,X_j=0; j \in \ell, \, j \neq i)P(\X+\lambda_t\Z_1=\mathbf{a},\X+\sigma \Z_2=\mathbf{b} \, | \, X_i=1,X_j=0; j \in \ell, \, j \neq i)}{P(\X+\lambda_t\Z_1=\mathbf{a},\X+\sigma\Z_2=\mathbf{b})}\\
&=\Big(\frac{1}{K}\Big)\frac{P(\X+\lambda_t\Z_1=\mathbf{a},\X+\sigma \Z_2=\mathbf{b} \, | \, X_i=1,X_j=0; j \in \ell,\,  j \neq i)}{P(\X+\lambda_t\Z_1=\mathbf{a},\X+\sigma\Z_2=\mathbf{b})}.
\end{aligned}
\end{equation}
WLOG, we assume that $i=1$, i.e., it is the first element in section $\ell$ and $j\in \{2, 3, \dots, K\}$. Now we simplify the two remaining probabilities in \eqref{eq:Blocksparse_denoiser}.  
First note that
\begin{equation}\label{eq:Blocksparse_prob1}
\begin{aligned}
&P(\X+\lambda_t\Z_1=\mathbf{a},\X+\sigma\Z_2=\mathbf{b} \, |\, X_1=1, X_2 = X_3 = \ldots = X_K =0)\\
&=P(\lambda_t \Z_1={\mathbf{a}- [1,0, \ldots, 0]^T}, \sigma \Z_2={\mathbf{b}-[1,0, \ldots, 0]^T})\\
&=P\left(Z_{11}=\frac{a_1-1}{\lambda_t}\right)P\left(Z_{21}=\frac{b_1-1}{\sigma}\right)\prod_{j=2}^{K} P\left(Z_{1j}=\frac{a_j}{\lambda_t}\right)P\left(Z_{2j}=\frac{b_j}{\sigma}\right)\\
&=\left(\frac{1}{2\pi {\lambda_t \sigma}}\right)^K \exp\left(\frac{-\lambda_t^2-\sigma^2}{2\lambda_t^2\sigma^2}\right)\exp\Big(-\sum_{j=1}^{K}\Big(\frac{a_j^2}{2\lambda_t^2}+\frac{b_j^2}{2\sigma^2}\Big)\Big)\exp\left(\frac{a_1}{\lambda_t^2}+\frac{b_1}{\sigma^2}\right).
\end{aligned}
\end{equation}
Similarly, we have
\begin{equation}\label{eq:Blocksparse_prob2}
\begin{aligned}
&P(\X+\lambda_t\Z_1=\mathbf{a},\X+\sigma\Z_2=\mathbf{b})\\
&=\sum_{k=1}^{K}P(X_k=1, X_j=0; j \in \ell, \, j \neq k) P(\X +\lambda_t\Z_1=\mathbf{a}, \X +\sigma \Z_2=\mathbf{b} \, | \, X_k=1, X_j=0; j \in \ell, \, j \neq k)\\
&= \frac{1}{K} \left(\frac{1}{2\pi\lambda_t \sigma}\right)^K \exp\left(\frac{-\lambda_t^2-\sigma^2}{2\lambda_t^2\sigma^2}\right)\exp\Big(-\sum_{j=1}^{K}\Big(\frac{a_j^2}{2\lambda_t^2}+\frac{b_j^2}{2\sigma^2}\Big)\Big)\sum_{k=1}^{K} \exp\left(\frac{a_k}{\lambda_t^2}+\frac{b_k}{\sigma^2}\right).
\end{aligned}
\end{equation}
Thus, combining \eqref{eq:Blocksparse_denoiser}-\eqref{eq:Blocksparse_prob2}, we find the desired result.
\end{proof}

\begin{proof}[Proof of Lemma \ref{lem:B2_proof}]
We will now appeal to Lemma \ref{lem:Lipschitz_non-separable} in order to show that the blockwise denoiser given in \eqref{eq:block_denoiser} is Lipschitz.  First, consider the partials w.r.t. $\mathbf{a}$,
\begin{equation*}
\begin{aligned}
&\left\lvert\frac{\partial}{\partial a_i} [g^{\ell}_t(\mathbf{a}, \mathbf{b})]_j \right \lvert
=\Big\lvert\frac{\partial}{\partial a_i}\Big( \frac{\exp\left(({a_j}/{\lambda_t^2})+({b_j}/{\sigma^2})\right)}{\sum_{k=1}^{K}\exp\left(({a_k}/{\lambda_t^2})+({b_k}/{\sigma^2})\right)} \Big)\Big \lvert \\
&= \frac{1}{\lambda_t^2}\Big( \frac{\exp\left(({a_j}/{\lambda_t^2})+({b_j}/{\sigma^2})\right)}{\sum_{k=1}^{K}\exp\left(({a_k}/{\lambda_t^2})+({b_k}/{\sigma^2})\right)} \Big) \Big\lvert\frac{\sum_{k=1}^{K}\exp\left(({a_k}/{\lambda_t^2})+({b_k}/{\sigma^2})\right)\mathbb{I}\{i=j\}-\exp\left(({a_i}/{\lambda_t^2})+ ({b_i}/{\sigma^2})\right) }{\sum_{k=1}^{K}\exp\left(({a_k}/{\lambda_t^2})+({b_k}/{\sigma^2})\right)}\Big\lvert.
\end{aligned}
\end{equation*}
We will show that
\begin{equation}
\frac{\exp\left(({a_j}/{\lambda_t^2})+({b_j}/{\sigma^2})\right)}{\sum_{k=1}^{K}\exp\left(({a_k}/{\lambda_t^2})+({b_k}/{\sigma^2})\right)} \Big\lvert\frac{\sum_{k=1}^{K}\exp\left(({a_k}/{\lambda_t^2})+({b_k}/{\sigma^2})\right)\mathbb{I}\{i=j\}-\exp\left(({a_i}/{\lambda_t^2})+ ({b_i}/{\sigma^2})\right) }{\sum_{k=1}^{K}\exp\left(({a_k}/{\lambda_t^2})+({b_k}/{\sigma^2})\right)}\Big\lvert \leq 1,
\label{eq:toshow1}
\end{equation}
in which case 
\[\left\lvert\frac{\partial}{\partial a_i} [g^{\ell}_t(\mathbf{a}, \mathbf{b})]_j \right \lvert\leq \frac{1}{\lambda_t^2}\leq \frac{1}{\sigma_w^2}.\]
Now we show \eqref{eq:toshow1}.  First note that \eqref{eq:toshow1} is obviously true if $i \neq j$, because
\[\frac{\exp\left(({a_j}/{\lambda_t^2})+ ({b_j}/{\sigma^2})\right)\exp\left(({a_i}/{\lambda_t^2})+ ({b_i}/{\sigma^2})\right)}{\left(\sum_{k=1}^{K}\exp\left(({a_k}/{\lambda_t^2})+({b_k}/{\sigma^2})\right)\right)^2}\leq1,\] 
so we consider the case $i=j$.  Then,
\begin{align*}
\frac{\exp\left(({a_i}/{\lambda_t^2})+({b_i}/{\sigma^2})\right)}{\sum_{k=1}^{K}\exp\left(({a_k}/{\lambda_t^2})+({b_k}/{\sigma^2})\right)} &\Big\lvert\frac{\sum_{k=1}^{K}\exp\left(({a_k}/{\lambda_t^2})+({b_k}/{\sigma^2})\right)-\exp\left(({a_i}/{\lambda_t^2})+ ({b_i}/{\sigma^2})\right) }{\sum_{k=1}^{K}\exp\left(({a_k}/{\lambda_t^2})+({b_k}/{\sigma^2})\right)}\Big\lvert \\
&=\frac{\exp\left(({a_i}/{\lambda_t^2})+ ({b_i}/{\sigma^2})\right) \sum_{k\neq i}\exp\left(({a_k}/{\lambda_t^2})+({b_k}/{\sigma^2})\right) }{\left(\exp\left(({a_i}/{\lambda_t^2})+ ({b_i}/{\sigma^2})\right) + \sum_{k\neq i}\exp\left(({a_k}/{\lambda_t^2})+({b_k}/{\sigma^2})\right)\right)^2}.
\end{align*}
The above is upper bounded by $1$ since the numerator is positive.

We note that the bound $\left\lvert\frac{\partial}{\partial b_i} [g^{\ell}_t(\mathbf{a}, \mathbf{b})]_j \right \lvert \leq \frac{1}{\sigma^2}$ can be shown similarly.

Now applying Lemma~\ref{lem:Lipschitz_non-separable}, we have shown that $[g^{\ell}_t(\mathbf{a}, \mathbf{b})]_j$ is Lipschitz with constant $\sqrt{\frac{2K^2}{\sigma_w^4} + \frac{2K^2}{\sigma^4}}.$
Therefore, for any index  $j\in \{1, 2, \ldots, K\}$ and $\ell \in \{ 1, 2, \ldots, L\}$,
\begin{equation}
    \label{eq:lemma_application1}
\Big\lvert [g^{\ell}_t(\mathbf{a}, \mathbf{b})]_j - [g^{\ell}_t(\widetilde{\mathbf{a}}, \widetilde{\mathbf{b}})]_j    \Big \lvert^2  \leq 2K^2 \Big(\frac{1}{\sigma_w^4} + \frac{1}{\sigma^4}\Big) \frac{ \lvert  \lvert (\mathbf{a}, \mathbf{b}) - (\widetilde{\mathbf{a}}, \widetilde{\mathbf{b}})  \lvert  \lvert^2}{2K} = K \Big(\frac{1}{\sigma_w^4} + \frac{1}{\sigma^4}\Big)  \lvert \lvert (\mathbf{a}, \mathbf{b}) - (\widetilde{\mathbf{a}}, \widetilde{\mathbf{b}})  \lvert\lvert^2.
\end{equation}
Now considering the denoiser in \eqref{eq:overall_denoiser}, we have
\begin{equation*}
\begin{split}
& \lvert  \lvert g_t(\mathbf{x}, \mathbf{y}) - g_t(\mathbf{x}, \mathbf{y})  \lvert \lvert^2 
= \sum_{\ell=1}^L \sum_{k=1}^K\Big \lvert [g^{\ell}_t(\mathbf{x}_{(\ell)}, \mathbf{y}_{(\ell)})]_k - [g^{\ell}_t(\widetilde{\mathbf{x}}_{(\ell)}, \widetilde{\mathbf{y}}_{(\ell)})]_k \Big \lvert^2 \\
&\leq \sum_{\ell=1}^L \sum_{k=1}^K K \Big(\frac{1}{\sigma_w^4} + \frac{1}{\sigma^4}\Big) \lvert   \lvert (\mathbf{x}_{(\ell)}, \mathbf{y}_{(\ell)}) - (\widetilde{\mathbf{x}}_{(\ell)}, \widetilde{\mathbf{y}}_{(\ell)})  \lvert  \lvert^2 =  K^2 \Big(\frac{1}{\sigma_w^4} + \frac{1}{\sigma^4}\Big)\lvert \lvert (\mathbf{x}, \mathbf{y}) - (\widetilde{\mathbf{x}}, \widetilde{\mathbf{y}})  \lvert  \lvert^2.
\end{split}
\end{equation*}

Therefore, we have the desired result:
\[\frac{\lvert \lvert g_t(\mathbf{x}, \mathbf{y}) - g_t(\mathbf{x}, \mathbf{y})  \lvert  \lvert}{\sqrt{n}}\leq \sqrt{2}K \Big(\frac{1}{\sigma_w^2} + \frac{1}{\sigma^2}\Big)\frac{ \lvert   \lvert (\mathbf{x}, \mathbf{y}) - (\widetilde{\mathbf{x}}, \widetilde{\mathbf{y}}) \lvert  \lvert}{\sqrt{2n}}.\]
\end{proof}

\begin{proof}[Proof of Lemma \ref{lem:B5_proof}]
We first study $\lim_{n\to\infty}\frac{1}{n} \mathbb{E}_{\Z}[\x^T g_t(\x+\Z,\widetilde{\x})]$.
Note that we have used the notation $\x_{(\ell)}$ to denote the $\ell^{th}$ block of $\x$.  Recall that $\x_{(\ell)}$ is a $K$-length vector, and we let $[\x_{(\ell)}]_k$ 
denote its $k^{th}$ entry for $k \in \{1, 2, \ldots, K\}$.  First, using the closed-form value of the blockwise denoisers given in Lemma \ref{lem:denoiser_simplify}, namely
\begin{equation*}
[g^{\ell}_t(\mathbf{a}, \mathbf{b})]_k =\frac{\exp\left((a_{k}/\lambda_t^2)+(b_k/\sigma^2)\right)}{\sum_{i=1}^{K}\exp\left(({a_i}/{\lambda_t^2})+({b_i}/{\sigma^2})\right)},
\end{equation*}
we have
\begin{equation}
\begin{aligned}
\lim_{n\to\infty}\frac{1}{n} \mathbb{E}_{\Z}[\x^T g_t(\x+\Z,\widetilde{\x})]&=\lim_{n\to\infty}\frac{1}{n} \sum_{\ell=1}^{L} \sum_{k=1}^K \mathbb{E}_{\Z}\left\{ [\x_{(\ell)}]_k  [g^{\ell}_t(\x+\Z,\widetilde{\x})]_k\right\}\\
&=\lim_{n\to\infty}\frac{1}{n} \sum_{\ell=1}^{L} \sum_{k=1}^K \mathbb{E}_{\Z}\left[\frac{[\x_{(\ell)}]_k \exp\left(([\x_{(\ell)}+\Z_{(\ell)}]_k/\lambda_t^2)+([\widetilde{\x}_{(\ell)}]_k/\sigma^2)\right)}{\sum_{i=1}^{K}\exp\left(({[\x_{(\ell)}+\Z_{(\ell)}]_i}/{\lambda_t^2})+([\widetilde{\x}_{(\ell)}]_i/{\sigma^2})\right)}\right]\\
&=\lim_{L \to\infty}\frac{1}{L} \sum_{\ell=1}^{L} \frac{1}{K} \mathbb{E}_{\Z}\left[\frac{\sum_{k=1}^K  [\x_{(\ell)}]_k \exp\left(([\x_{(\ell)}+\Z_{(\ell)}]_k/\lambda_t^2)+([\widetilde{\x}_{(\ell)}]_k/\sigma^2)\right)}{\sum_{i=1}^{K}\exp\left(({[\x_{(\ell)}+\Z_{(\ell)}]_i}/{\lambda_t^2})+([\widetilde{\x}_{(\ell)}]_i/{\sigma^2})\right)}\right].
\label{eq:first_bound}
\end{aligned}
\end{equation}
 
Now we consider the expectation
\[\mathbb{E}_{\Z}\left[\frac{\sum_{k=1}^K  [\x_{(\ell)}]_k \exp\left(([\x_{(\ell)}+\Z_{(\ell)}]_k/\lambda_t^2)+([\widetilde{\x}_{(\ell)}]_k/\sigma^2)\right)}{\sum_{i=1}^{K}\exp\left(({[\x_{(\ell)}+\Z_{(\ell)}]_i}/{\lambda_t^2})+([\widetilde{\x}_{(\ell)}]_i/{\sigma^2})\right)}\right].\]
We note that the $K$-length vector $\x_{(\ell)}$ has a single non-zero value taking the value one.  WLOG assume that $[\x_{(\ell)}]_1 =1$ and $[\x_{(\ell)}]_2 = [\x_{(\ell)}]_3 = \ldots [\x_{(\ell)}]_K =0$.  Then,
\begin{equation*}
    \begin{aligned}
&\mathbb{E}_{\Z}\left[\frac{\sum_{k=1}^K  [\x_{(\ell)}]_k \exp\left(([\x_{(\ell)}+\Z_{(\ell)}]_k/\lambda_t^2)+([\widetilde{\x}_{(\ell)}]_k/\sigma^2)\right)}{\sum_{i=1}^{K}\exp\left(({[\x_{(\ell)}+\Z_{(\ell)}]_i}/{\lambda_t^2})+([\widetilde{\x}_{(\ell)}]_i/{\sigma^2})\right)}\right] \\
&= \mathbb{E}_{\Z}\left[\frac{ \exp\left((1/\lambda_t^2) + ( [\Z_{(\ell)}]_1/\lambda_t^2) +([\widetilde{\x}_{(\ell)}]_1/\sigma^2)\right)}{ \exp\left((1/\lambda_t^2) + ( [\Z_{(\ell)}]_1/\lambda_t^2)+([\widetilde{\x}_{(\ell)}]_1/\sigma^2)\right) + \sum_{i=2}^{K}\exp\left(({[\Z_{(\ell)}]_i}/{\lambda_t^2})+([\widetilde{\x}_{(\ell)}]_i/{\sigma^2})\right)}\right]\\
&= \mathbb{E}_{\Z}\left[\frac{ \exp\left((1/\lambda_t^2) + ( [\Z_{(\ell)}]_1/\lambda_t^2) +(1/\sigma^2) + ([\widetilde{\z}_{(\ell)}]_1/\sigma^2)\right)}{ \exp\left((1/\lambda_t^2) + ( [\Z_{(\ell)}]_1/\lambda_t^2)+(1/\sigma^2) + ([\widetilde{\z}_{(\ell)}]_1/\sigma^2)\right) + \sum_{i=2}^{K}\exp\left(({[\Z_{(\ell)}]_i}/{\lambda_t^2})+([\widetilde{\z}_{(\ell)}]_i/{\sigma^2})\right)}\right],
\end{aligned}
\end{equation*}
where in the last step we have used that in our model, $\widetilde{\x} = \x+\widetilde{\z} = \x + \mathcal{N}(0, \sigma^2 \mathbb{I}_n)$.  Now define the RV (where the randomness is w.r.t. realizations from a RV $\widetilde{\Z}$),
\[Y_{\ell} : =  \frac{1}{K}\mathbb{E}_{\Z}\left[\frac{  \exp\left((1/\lambda_t^2) + ( [\Z_{(\ell)}]_1/\lambda_t^2) +(1/\sigma^2) + ([\widetilde{\z}_{(\ell)}]_1/\sigma^2)\right)}{ \exp\left((1/\lambda_t^2) + ( [\Z_{(\ell)}]_1/\lambda_t^2)+(1/\sigma^2) + ([\widetilde{\z}_{(\ell)}]_1/\sigma^2)\right) + \sum_{i=2}^{K}\exp\left(({[\Z_{(\ell)}]_i}/{\lambda_t^2})+([\widetilde{\z}_{(\ell)}]_i/{\sigma^2})\right)}\right]. \]
Plugging this into \eqref{eq:first_bound},
\begin{equation*}
\begin{aligned}
\lim_{n\to\infty}\frac{1}{n} \mathbb{E}_{\Z}[\x^T g_t(\x+\Z,\widetilde{\x})]
=\lim_{L\to\infty}\frac{1}{L} \sum_{\ell=1}^{L} Y_{\ell} = \mathbb{E}_{\widetilde{\Z}}[Y_{1}].
\end{aligned}
\end{equation*}
We can show that the above limit exists using the SLLN (stated in Thm \ref{thm:LLN}), since the $Y_{\ell}$ are i.i.d.\ for $\ell = 1, 2, \ldots, L$ with expectation $\mathbb{E}_{\widetilde{\Z}}[Y_{1}].$  It is not hard to show that $\mathbb{E}_{\widetilde{\Z}}[Y_{1}] < \infty.$

The second equation in $\textbf{(B5)}$ can be shown similarly, 
\begin{equation}
\begin{aligned}
&\lim_{n \to\infty}\frac{1}{n} \mathbb{E}_{\Z,\Z'}\left[g_t(\x+ \Z,\widetilde{\x})^T
g_s(\x+ \Z',\widetilde{\x})\right]\\
&=\lim_{n\to\infty}\frac{1}{n} \sum_{\ell=1}^{L} \sum_{k=1}^K \mathbb{E}_{\Z,\Z'}\left[[g^{\ell}_t(\x+ \Z,\widetilde{\x})]_k [g^{\ell}_s(\x+ \Z',\widetilde{\x})]_k \right]\\
&=\lim_{n\to\infty}\frac{1}{n} \sum_{\ell=1}^{L} \sum_{k=1}^K \mathbb{E}_{\Z,\Z'}\left[\frac{\exp\left(([\x_{(\ell)}+\Z_{(\ell)}]_k/{\lambda_t^2})+({[\widetilde{\x}_{(\ell)}]_k}/{\sigma^2})\right) \exp\left(([\x_{(\ell)}+\Z'_{(\ell)}]_k/{\lambda_t^2})+({[\widetilde{\x}_{(\ell)}]_k}/{\sigma^2})\right)}{\sum_{i=1}^{K}\exp\left(([\x_{(\ell)}+\Z_{(\ell)}]_i/{\lambda_t^2})+({[\widetilde{\x}_{(\ell)}]_i}/{\sigma^2})\right)\sum_{j=1}^{K}\exp\left(([\x_{(\ell)}+\Z'_{(\ell)}]_j/{\lambda_t^2})+({[\widetilde{\x}_{(\ell)}]_j}/{\sigma^2})\right)}\right]\\
&=\lim_{L\to\infty}\frac{1}{L} \sum_{\ell=1}^{L}\frac{1}{K}  \mathbb{E}_{\Z,\Z'}\left[\frac{\sum_{k=1}^K \exp\left(([\x_{(\ell)}+\Z_{(\ell)}]_k/{\lambda_t^2})+({[\widetilde{\x}_{(\ell)}]_k}/{\sigma^2})\right) \exp\left(([\x_{(\ell)}+\Z'_{(\ell)}]_k/{\lambda_t^2})+({[\widetilde{\x}_{(\ell)}]_k}/{\sigma^2})\right)}{\sum_{i=1}^{K}\exp\left(([\x_{(\ell)}+\Z_{(\ell)}]_i/{\lambda_t^2})+({[\widetilde{\x}_{(\ell)}]_i}/{\sigma^2})\right)\sum_{j=1}^{K}\exp\left(([\x_{(\ell)}+\Z'_{(\ell)}]_j/{\lambda_t^2})+({[\widetilde{\x}_{(\ell)}]_j}/{\sigma^2})\right)}\right].
\label{eq:B5_limit2}
\end{aligned}
\end{equation}
Again, we assume WLOG that $[\x_{(\ell)}]_1 =1$ and $[\x_{(\ell)}]_2 = [\x_{(\ell)}]_3 = \ldots [\x_{(\ell)}]_K =0$. Consider the numerator inside the expectation on the right side of \eqref{eq:B5_limit2}.  We have
\begin{equation}
\begin{split}
&\sum_{k=1}^K \exp\left(\frac{[\x_{(\ell)}+\Z_{(\ell)}]_k}{\lambda_t^2}+\frac{[\widetilde{\x}_{(\ell)}]_k}{\sigma^2}\right) \exp\left(\frac{[\x_{(\ell)}+\Z'_{(\ell)}]_k}{\lambda_t^2}+\frac{[\widetilde{\x}_{(\ell)}]_k}{\sigma^2}\right) \\
&=  \exp\left(\frac{1 + [\Z_{(\ell)}]_1}{\lambda_t^2}+ \frac{1+[\Z'_{(\ell)}]_1}{\lambda_t^2}+\frac{2[\widetilde{\x}_{(\ell)}]_1}{\sigma^2}\right) + \sum_{k=2}^K \exp\left(\frac{[\Z_{(\ell)}]_k}{\lambda_t^2}+\frac{[\Z'_{(\ell)}]_k}{\lambda_t^2}+\frac{2[\widetilde{\x}_{(\ell)}]_k}{\sigma^2}\right)\\
&=  \exp\left( \frac{1+[\Z_{(\ell)}]_1}{\lambda_t^2}+ \frac{1+[\Z'_{(\ell)}]_1}{\lambda_t^2}+\frac{2(1+[\widetilde{\z}_{(\ell)}]_1)}{\sigma^2}\right) + \sum_{k=2}^K \exp\left(\frac{[\Z_{(\ell)}]_k}{\lambda_t^2}+\frac{[\Z'_{(\ell)}]_k}{\lambda_t^2}+\frac{2[\widetilde{\z}_{(\ell)}]_k}{\sigma^2}\right),
\label{eq:num}
\end{split}
\end{equation}
where in the last step we have used that in our model, $\widetilde{\x} = \x+\widetilde{\z} = \x + \mathcal{N}(0, \sigma^2 \mathbb{I})$.  Similarly for the denominator in \eqref{eq:B5_limit2},
\begin{equation}
\begin{split}
&\sum_{i=1}^{K}\exp\left(\frac{[\x_{(\ell)}+\Z_{(\ell)}]_i}{\lambda_t^2}+\frac{[\widetilde{\x}_{(\ell)}]_i}{\sigma^2}\right)\sum_{j=1}^{K}\exp\left(\frac{[\x_{(\ell)}+\Z'_{(\ell)}]_j}{\lambda_t^2}+\frac{[\widetilde{\x}_{(\ell)}]_j}{\sigma^2}\right)\\
&=\left(\exp\left(\frac{1+[\Z_{(\ell)}]_1}{\lambda_t^2}+\frac{1+[\widetilde{\Z}_{(\ell)}]_1}{\sigma^2}\right) + \sum_{i=2}^{K}\exp\left(\frac{[\Z_{(\ell)}]_i}{\lambda_t^2}+\frac{[\widetilde{\z}_{(\ell)}]_i}{\sigma^2}\right)\right)\times \\
&\qquad \left(\exp\left(\frac{1+[\Z'_{(\ell)}]_1}{\lambda_t^2}+\frac{1+[\widetilde{\z}_{(\ell)}]_1}{\sigma^2}\right) + \sum_{j=2}^{K}\exp\left(\frac{[\Z'_{(\ell)}]_j}{\lambda_t^2}+\frac{[\widetilde{\z}_{(\ell)}]_j}{\sigma^2}\right)\right).
\label{eq:denom}
\end{split}
\end{equation}
Plugging \eqref{eq:num} and \eqref{eq:denom} into \eqref{eq:B5_limit2}, 
\begin{equation}
\begin{aligned}
&\lim_{n \to\infty}\frac{1}{n} \mathbb{E}_{\Z,\Z'}\left[g_t(\x+ \Z,\widetilde{\x})^T
g_s(\x+ \Z',\widetilde{\x})\right]\\
&=\lim_{L\to\infty}\frac{1}{LK} \sum_{\ell=1}^{L}  \mathbb{E}_{\Z,\Z'}\left[\frac{e^{ \frac{1+[\Z_{(\ell)}]_1}{\lambda_t^2}+ \frac{1+[\Z'_{(\ell)}]_1}{\lambda_t^2}+\frac{2(1+[\widetilde{\Z}_{(\ell)}]_1)}{\sigma^2}} + \sum_{k=2}^K e^{\frac{[\Z_{(\ell)}]_k}{\lambda_t^2}+\frac{[\Z'_{(\ell)}]_k}{\lambda_t^2}+\frac{2[\widetilde{\Z}_{(\ell)}]_k}{\sigma^2}}}{\Big[e^{\frac{1+[\Z_{(\ell)}]_1}{\lambda_t^2}+\frac{1+[\widetilde{\Z}_{(\ell)}]_1}{\sigma^2}} + \sum_{i=2}^{K}e^{\frac{[\Z_{(\ell)}]_i}{\lambda_t^2}+\frac{[\widetilde{\Z}_{(\ell)}]_i}{\sigma^2}}\Big]\Big[e^{\frac{1+[\Z'_{(\ell)}]_1}{\lambda_t^2}+\frac{1+[\widetilde{\Z}_{(\ell)}]_1}{\sigma^2}} + \sum_{j=2}^{K}e^{\frac{[\Z'_{(\ell)}]_j}{\lambda_t^2}+\frac{[\widetilde{\Z}_{(\ell)}]_j}{\sigma^2}}\Big]}\right],
\label{eq:final_limit}
\end{aligned}
\end{equation}
and as before we define a RV
\[W_{\ell}:= \frac{1}{K}  \mathbb{E}_{\Z,\Z'}\left[\frac{e^{ \frac{1+[\Z_{(\ell)}]_1}{\lambda_t^2}+ \frac{1+[\Z'_{(\ell)}]_1}{\lambda_t^2}+\frac{2(1+[\widetilde{\Z}_{(\ell)}]_1)}{\sigma^2}} + \sum_{k=2}^K e^{\frac{[\Z_{(\ell)}]_k}{\lambda_t^2}+\frac{[\Z'_{(\ell)}]_k}{\lambda_t^2}+\frac{2[\widetilde{\Z}_{(\ell)}]_k}{\sigma^2}}}{\Big[e^{\frac{1+[\Z_{(\ell)}]_1}{\lambda_t^2}+\frac{1+[\widetilde{\Z}_{(\ell)}]_1}{\sigma^2}} + \sum_{i=2}^{K}e^{\frac{[\Z_{(\ell)}]_i}{\lambda_t^2}+\frac{[\widetilde{\Z}_{\ell}]_i}{\sigma^2}}\Big]\Big[e^{\frac{1+[\Z'_{(\ell)}]_1}{\lambda_t^2}+\frac{1+[\widetilde{\Z}_{(\ell)l}]_1}{\sigma^2}} + \sum_{j=2}^{K}e^{\frac{[\Z'_{(\ell)}]_j}{\lambda_t^2}+\frac{[\widetilde{\Z}_{(\ell)}]_j}{\sigma^2}}\Big]}\right],\]
where the randomness is w.r.t. $\widetilde{\Z}$.  Then, from \eqref{eq:final_limit},
\begin{equation*}
\begin{aligned}
&\lim_{n \to\infty}\frac{1}{n} \mathbb{E}_{\Z,\Z'}\left[g_t(\x+ \Z,\widetilde{\x})^T
g_s(\x+ \Z',\widetilde{\x})\right] = \lim_{L\to\infty}\frac{1}{L} \sum_{\ell=1}^{L} W_{\ell} = \mathbb{E}_{\widetilde{\Z}}[W_{1}].
\end{aligned}
\end{equation*}
We can show that the above limit holds due to the SLLN since $W_\ell$ are i.i.d.\ for $\ell = 1, 2, \ldots, L$ having expectation $\mathbb{E}_{\widetilde{\Z}}[W_{1}]$.  It is not hard to show that $\mathbb{E}_{\widetilde{\Z}}[W_{1}]< \infty$.

\end{proof}

\end{document}